\documentclass{article} 
\usepackage{spconf}
\usepackage{pgfplots}
\usepgfplotslibrary{groupplots}
\pgfplotsset{compat=1.18}
\pgfplotsset{scaled y ticks=false}
\usepackage{tikz}
\usepackage{pgfplots}
\usepackage{pgfplotstable}
\usepackage{subcaption}
\usepackage{enumitem}
\usetikzlibrary{positioning}
\usetikzlibrary{calc}
\usepackage{graphicx}
\definecolor{OIblue}{RGB}{0,114,178}
\definecolor{OIorange}{RGB}{230,159,0}
\definecolor{OIgreen}{RGB}{0,158,115}
\definecolor{OIvermillion}{RGB}{213,94,0}
\definecolor{OIsky}{RGB}{86,180,233}
\definecolor{OIpurple}{RGB}{204,121,167}

\pgfplotsset{table/search path={..,../csv_files_for_plots,csv_files_for_plots,.}}

\usepackage{amsmath}

\usepackage{amssymb}
\usepackage{amsfonts}
\usepackage{dsfont}
\usepackage{accents}
\usepackage{bm}
\usepackage{hyperref}
\usepackage{bbm}
\usepackage{amsthm}

\usepackage{tabularx} 
\usepackage{caption}
\usepackage{multirow}
\usepackage{booktabs}
\usepackage{graphicx}
\graphicspath{{Figures/}}
 \usepackage{subcaption}
 \usepackage{tcolorbox}
 \usepackage{wasysym}
  \usepackage{longtable}

\usepackage{tikz,pgfplots}
\usepgfplotslibrary{groupplots}
\usetikzlibrary{intersections,calc,arrows,matrix,spy}
\usepackage{color}
\definecolor{darkred}{rgb}{0.6,0,0}
\definecolor{darkgreen}{rgb}{0,0.5,0}
\definecolor{darkblue}{rgb}{0,0,0.5}
\definecolor{SkyBlue}{rgb}{0.53, 0.81, 0.92}
\pgfplotsset{compat=1.5.1}

\usepackage{algorithm}
\usepackage{algpseudocode}

\usepackage[normalem]{ulem}
\newcommand{\stkout}[1]{\ifmmode\text{\sout{\ensuremath{#1}}}\else\sout{#1}\fi}

\def\R{\mathbb{R}}          									 	          
      \def\Am{\mathbf{A}}     									 	          
\def\N{\mathbb{N}}          									    	      
\def\E{\mathbb{E}}

\newcommand{\xb}{\mathrm{\mathbf{x}}}
\newcommand{\yb}{\mathrm{\mathbf{y}}}

\newcommand{\Ab}{\mathrm{\mathbf{A}}}

\newcommand{\mub}{\mathrm{\bm{\mu}}}

\newcommand{\Lc}{\mathcal{L}}

\newcommand{\Xc}{\mathcal{X}}

\usepackage{tcolorbox}
\usepackage{mdframed}
\usepackage{lipsum}

\newtheorem{remark}{Remark}[section]

\newtheorem{lemma}{Lemma}[section]
\newtheorem{proposition}[lemma]{Proposition}

\def\1{1\!\!1}

\newcommand{\EE}{\mathbb{E}}
\newcommand{\asym}{\asymp}
\graphicspath{{./images/}}

\title{Gaussian Surrogates for Poisson Imaging: \\ Some Theoretical and Empirical Results}
\name{Alexandra Spitzer*, Lorenzo Baldassari*, Valentin Debarnot$^\dagger$, Ivan Dokmanić*}
\address{*Department of Mathematics and Computer Science, University of Basel, 4051 Basel, Switzerland \\ $^\dagger$INSA‐Lyon, Universite Claude Bernard Lyon 1, CNRS, Inserm, CREATIS UMR 5220, U1294}

\begin{document}
	
\maketitle

\begin{abstract}
In imaging inverse problems with Poisson-distributed measurements, it is common to use objectives derived from the Poisson likelihood. But performance is often evaluated by mean squared error (MSE), which raises a practical question: how much does a Poisson objective matter for MSE, even at low dose?
We analyze the MSE of Poisson and Gaussian surrogate reconstruction objectives under Poisson noise. In a stylized diagonal model, we show that the unregularized Poisson maximum-likelihood estimator can incur large MSE at low dose, while Poisson MAP mitigates this instability through regularization. We then study two Gaussian surrogate objectives: a heteroscedastic quadratic objective motivated by the normal approximation of Poisson data, and a homoscedastic quadratic objective that yields a simple linear estimator. We show that both surrogates can achieve MSE comparable to  Poisson MAP in the low-dose regime, despite departing from the Poisson likelihood. Numerical computed tomography experiments indicate that these conclusions extend beyond the stylized setting of our theoretical analysis.

\end{abstract}

\section{Introduction}
When photon or electron counts are low, the Poisson noise model is standard \cite{willett2007multiscale,bertero2009image,hohage2016inverse}  and it is natural to optimize a Poisson likelihood instead of a Gaussian quadratic data term for reconstruction \cite{fessler1995penalized,dey2006richardson,debarnot2019multiview}. But likelihood and reconstruction error are different objectives. Maximizing the correct likelihood does not guarantee minimal MSE---particularly in ill-posed problems where small eigenvalues amplify noise. This raises a practical question: how much MSE do we lose (or perhaps gain) by using objectives based on the Gaussian likelihood when the data are Poisson?
At high dose the Poisson distribution converges to a Gaussian, but we are interested in understanding what happens at low dose. We show that even in this regime, properly regularized Gaussian surrogates can be surprisingly competitive with Poisson objectives in terms of MSE. In a diagonal model that allows for closed-form analysis, we show that the unregularized Poisson MLE incurs large MSE at low dose: small diagonal entries push individual modes into an effective low-count regime where single-photon events create variance spikes (Section~\ref{sec:pois_mle}).

We then discuss two Gaussian surrogates. First, we consider a heteroscedastic objective motivated by the normal approximation of the Poisson  distribution. While our analysis shows that it can achieve smaller MSE than the Poisson MLE (Proposition~\ref{prop:per-mode-mse-hgmle}), it still retains some of the practical difficulties of Poisson modeling, since the log-likelihood remains non-quadratic. We then study a homoscedastic Gaussian surrogate (Section~\ref{sec:OG}) that departs entirely from Poisson statistics. We show that it yields provably smaller MSE in the same low-dose regime, and we compare it with Poisson maximum a posteriori (MAP) with Tikhonov regularization (Section~\ref{sec:pois_tik}). The homoscedastic surrogate is particularly attractive: its quadratic objective yields a linear estimator, making it both computationally simple and analytically tractable.

Numerical experiments on computed tomography (Section~\ref{sec:experiments}) confirm that these findings extend beyond the diagonal setting. We find that simple regularized quadratic solvers, including ordinary least squares, match Poisson MAP in MSE across all tested count levels. 

\subsection{Related works}
Poisson noise arises in photon-limited imaging modalities, including fluorescence deblurring, emission tomography, astronomical imaging, and denoising \cite{mcnally1999three,harmany2011spiral,krishnamurthy2010multiscale}.
The Poisson MLE is often computed using the Richardson–Lucy algorithm \cite{richardson1972bayesian,lucy1974iterative}, which corresponds to ML-EM in emission tomography \cite{shepp1982maximum,lange1984reconstruction}.
In practice, ML-EM reconstructions are regularized to mitigate variance explosion, for example via early stopping or explicit penalties \cite{veklerov2007stopping}. Poisson MAP is commonly tackled via one-step-late (OSL) MAP-EM updates \cite{green1990bayesian}. OSL is widely used in tomography, including with TV-type regularization \cite{panin1998total, Sawatzky}.

An alternative line of work relies on variance-stabilizing transformations, most notably the Anscombe transform \cite{anscombe1948transformation}, which approximates Poisson noise as additive Gaussian noise with nearly constant variance. 
These approximations lead to weighted least squares (WLS) and penalized WLS 
formulations \cite{curtis1975simple,stagliano2011analysis,singh2025learning,gribonval2021bayesian}, which are widely used in tomography.

\section{A stylized theoretical analysis}\label{SEC:METHODS}
We consider measurements $\yb= (y_1,\hdots, y_{m})\in\N^{m}$ modeled as independent Poisson counts
\begin{equation}\label{eq:obs_general}
y_j \sim \mathcal P\big((s\Am^{}\xb^{\star})_j\big),\qquad j=1,\ldots,m,
\end{equation}
where $\mathcal P(\lambda)$ denotes the Poisson distribution with parameter $\lambda$, $\Am:\R^{n} \to \R^{m}$ is a forward operator such that  $(\Am\xb^\star)_j\geq 0$ for all $j=1,\ldots,m$, $\xb^{\star}\in\R^{n}$ is the unknown signal, and $s>0$ is the dose. 
A key quantity in our analysis is the expected number of counts in measurement $j$:
\begin{equation}\label{eq:effective_counts_general}
\mu_j \;:=\; \mathbb{E} [y_j] = s\,(\Am \xb^{\star})_j,\qquad j=1,\ldots,m.
\end{equation}
In this paper, we focus on the low dose regime: $\mu_j\ll 1$.
We aim to characterize the mean-squared error (MSE),
\begin{equation}\label{eq:mse_target}
\mathrm{MSE}(\hat{\xb})
\;:=\;
\E\big\|\hat{\xb}-\xb^{\star} \big\|_2^2.
\end{equation}
In particular, we study how the MSE depends on the dose $s$,
the operator $\Am$,  and, importantly, the reconstruction strategy. 

\subsection{Poisson MLE}\label{sec:pois_mle}

We begin by considering the Poisson MLE to  estimate $\xb^{\star}$ in the low-dose regime ($\mu_j\ll 1$ for many $j$), a straightforward starting point since it is derived from the Poisson model:
\begin{equation}\label{eq:MLE_IP}
\xb_{\mathrm{MLE,P}} \; \in \; \arg\min_{\xb\in \mathcal X_+} \mathcal L_P(\xb;\yb),
\end{equation}
where $\Xc_+ := \{ \xb \in \mathbb R^n_+ :(\Am \xb)_j \geq 0, \, j=1,\ldots,m\}$ and
\begin{equation}\label{def:P-likelihood}
\mathcal L_P(\xb;\yb)
\; := \; \sum_{j=1}^{m} \left( s(\Am\xb)_j
- y_j\log\!\left(s(\Am\xb)_j\right) \right).
\end{equation}
We interpret $\mathcal L_P(\xb;\yb)$ as an extended-valued function on $\mathcal X_+$ by setting $\mathcal L_P(\xb;\yb)=+\infty$ whenever $s(\Am\xb)_j=0$ for some $j$ with $y_j>0$.
Moreover, we adopt the convention $0\log 0:=0$, i.e.,
$y_j\log\!\big(s(\Am\xb)_j\big):=0$ when $y_j=0$ and $s(\Am\xb)_j=0$.

To analyze the MSE of $\xb_{\mathrm{MLE, P}}$, we study a stylized model in which $\Am$ is linear and diagonal (and $m=n$). In the Gaussian setting, many linear inverse problems can be reduced to diagonal form without breaking Gaussianity. This is not true for the Poisson random variables, which are not preserved by a linear change of coordinates. Nonetheless, the diagonal model still provides useful insight, as it allows an explicit analysis of how the MSE depends on the dose  $s$ and the conditioning of $\Am$. We leave the much more involved non-diagonal case to future work.

We consider the setting where only the first $d \leq m$ components of the measurement vector are observed:\begin{equation}\label{eq:diag_obs}
y_j\sim\mathcal P(\mu_j),\qquad \mu _j:=s\,a_j\,x_j^{\star},\qquad j=1,\dots,d.
\end{equation}
The remaining $m-d$ components are unobserved. We also assume that $a_j > 0$ for all $j =1,\ldots, d$.

We interpret $d$ as the \emph{resolution} of the problem: for convolution operators, for example, diagonalization is achieved in the Fourier basis, and truncating to $d$ modes corresponds to a frequency cutoff, dictated, for example, by the sensor size; larger $d$ means a higher cutoff (larger effective bandwidth) and thus the ability to represent finer spatial scales.
The choice of $d$ will be important in our analysis, since it governs how the low-dose regime manifests across modes; we make this precise below.

Note that in the diagonal setting, the Poisson MLE admits an explicit closed-form expression
\begin{equation}\label{eq:pois_mle_diag}
\hat x_{\mathrm{MLE,P},j}=\frac{y_j}{s a_j},\qquad j\le d,
\end{equation}
and its MSE satisfies
\begin{equation}\label{eq:mse_biasvar}
\EE\|\hat \xb_{\mathrm{MLE,P}}-\xb^\star\|_2^2
=
\underbrace{\sum_{j=1}^d\frac{x_j^{\star}}{s a_j}}_{\text{estimation variance}}
\;+\;
\underbrace{\sum_{j=d+1}^m(x_j^\star)^2}_{\text{truncation bias}}.
\end{equation}
By making explicit the trade-off between estimation variance and truncation error, it highlights a key fact: small $s a_j$  can make the MSE variance-dominated. This leads to two remarks, in which the role of the resolution $d$ becomes explicit.

\textbf{Remark 1: ill-posedness.} 
In ill-posed inverse problems, one has $a_j\downarrow 0$ as the resolution increases. Hence, for large $d$, there are indices $j$ for which $s a_j$ is small even for moderate dose $s$, so the variance term in \eqref{eq:mse_biasvar} can dominate the MSE.

\textbf{Remark 2: MSE-based resolution choice.} For $s$ not too small, one can select $d$ by balancing the estimation-variance and truncation errors in \eqref{eq:mse_biasvar}; this yields a principled, dose-dependent notion of \emph{MSE-optimal resolution}.
Assume, for example, that the diagonal entries of $\Am$ and the signal coefficients $\xb^\star$ decay polynomially,
\begin{equation}\label{eq:poly_scales_methods}
a_j\asym j^{-\beta},\qquad x_j^\star\asym j^{-\alpha},
\qquad \beta>0,\ \alpha>\tfrac12.
\end{equation}
The condition $\alpha>\tfrac12$ ensures that  $\xb^\star \in \ell_2$ (and hence has finite energy) even as the dimension $m$ increases.
Balancing the two terms in \eqref{eq:mse_biasvar} yields the dose-dependent resolution
\begin{equation}\label{eq:d_of_s_methods}
d(s)\asym s^{1/(\alpha+\beta)}.
\end{equation}
Note that $d(s)$ also captures the onset of the low-dose regime: the index $j_\star(s)$ at which $\mu_j$ transitions from $\gg1$ to $\lesssim 1$ scales precisely as  $j_\star(s)\asym s^{1/(\alpha+\beta)}$.

This model reveals an important phenomenon for the MSE of the Poisson MLE: $d(s)$ increases with dose, but ill-posedness limits how far it can be pushed before the problem enters a low-dose regime where variance dominates. 
For extreme low-dose measurements $s \ll 1$, this cannot be avoided by tuning $d(s)$: the expected counts $\mu_j = s a_j x_j^\star\ll 1$ across all indices.
These considerations motivate the shrinkage strategies discussed next.

\subsection{Surrogate models}
In the low dose regime ($\mu_j\ll 1$), the probability of observing $y_j\geq 2$ is of order $\mu_j^2$. 
In particular, $y_j=0$ with high probability, while on the event $y_j=1$ the Poisson MLE produces a spike $\hat x_{\mathrm{MLE,P},j}=1/(s a_j)$ (see \eqref{eq:pois_mle_diag}).
These spikes can be large and dominate the total MSE as discussed above. This motivates reconstruction strategies that regularize the low-dose modes.
In what follows, we analyze two such strategies:
(i) a Poisson-likelihood objective with Tikhonov regularization, and (ii) its Gaussian counterpart.
We show that, in the low-dose regime, the Gaussian surrogate can achieve an explicit reduction in the MSE comparable to that of Poisson MAP---which may seem counterintuitive given that the measurements are Poisson distributed.

\subsubsection{Poisson MAP}\label{sec:pois_tik}

In practice, the Poisson MLE is \emph{rarely maximized to convergence}; instead, it is  implicitly regularized by early stopping an expectation–maximization algorithm. A similar effect can be achieved by explicit regularization. For example, the maximum a posteriori (MAP) estimator under an isotropic Gaussian prior results in Tikhonov regularization,
\begin{equation}\label{eq:pois_tikh}
\hat \xb_{\mathrm{Tik,P}} \in
\arg\min_{\xb\in \mathcal X_+}\Big\{
\Lc_P(\xb;\yb) + \frac{\tau}{2}\|\xb\|_2^2
\Big\},
\qquad \tau>0.
\end{equation}
In the diagonal model \eqref{eq:diag_obs}, Problem \eqref{eq:pois_tikh} decouples across coordinates; the solution has the closed form
\begin{equation}\label{eq:tik_closed_short}
\hat x_{\mathrm{Tik,P},j}(y_j)=\begin{cases}
\displaystyle 
\frac{-s a_j + \sqrt{(s a_j)^2 + 4\tau y_j}}{2\tau}& j\leq d,\\
\displaystyle 0, & j >d.
\end{cases} 
\end{equation}
The following proposition quantifies how Tikhonov regularization can mitigate the spikes that blow up the variance term in the MSE. Relative to the unregularized Poisson MLE, the MSE improves by a \emph{mode-dependent} factor.

\begin{proposition}\label{prop:per-mode-mse-pmle}
Define the \emph{effective regularization level}
\begin{equation}\label{eq:gamma_short}
\gamma_j:=\frac{\tau}{(s a_j)^2}, \qquad j\leq d.
\end{equation}
In the low-dose limit, as $\mu_j=s a_j x_j^\star\to 0$, the per-mode MSE ratio relative to the Poisson MLE satisfies
\begin{equation}\label{eq:tik_ratio_short}
\frac{\EE(\hat x_{\mathrm{Tik,P},j}-x_j^\star)^2}{\EE(\hat x_{\mathrm{MLE,P},j}-x_j^\star)^2}
=
\left(\frac{2}{1+\sqrt{1+4\gamma_j}}\right)^2+ O(\mu_j).
\end{equation}
The notation $f(\mu_j)= O(\mu_j)$ means that there exist  constants $C,\mu_0>0$ such that $|f(\mu_j)|\leq C\mu_j$ for all $0<\mu_j < \mu_0$.
\end{proposition}
A similar expression for the \emph{global} MSE ratio (under a uniform low-dose condition over resolved modes) is provided in the Supplementary Materials~\ref{sec:global-mse-tik}.

Note that the improvement factor in \eqref{eq:tik_ratio_short} depends on $\gamma_j=\tau/(s a_j)^2$:
as $a_j\downarrow 0$, $\gamma_j$ increases and the shrinkage becomes stronger. This is consistent with the role of Tikhonov regularization: it shrinks most strongly the modes that would otherwise exhibit the largest low-dose variance growth.

\subsubsection{Homoscedastic Gaussian MAP}\label{sec:OG}
Poisson objectives can be difficult to optimize and analyze, because the log-likelihood includes the term $\log \big(s(\Am \xb)_j\big)$. It is therefore natural to ask whether simple quadratic losses can yield similar or even better MSE.
In high-dose regimes, Gaussian surrogates are motivated by the normal approximation: if $X\sim\mathcal P(\lambda)$, then as $\lambda\to\infty$, $(X-\lambda)/\sqrt{\lambda}$ converges in distribution to $\mathcal N(0,1)$ \cite{stagliano2011analysis}.
Applied to \eqref{eq:obs_general}, this suggests using the \emph{heteroscedastic} (non-constant variance across measurements) Gaussian approximation
$y_j \approx \mathcal N\!\big(s(\Am\xb^\star)_j,\; s(\Am\xb^\star)_j\big)$, which one might hope to use even at low dose.
This choice, while principled, does not eliminate all Poisson-specific difficulties: the associated negative log-likelihood is still non-quadratic, with logarithmic and reciprocal terms in $s(\Am\xb)_j$, and is therefore less straightforward to handle computationally and analytically. In Supplementary Materials~\ref{sec:HG}, we analyze this model and show that, despite being a Gaussian approximation of the Poisson distribution, it can achieve a smaller MSE than the Poisson MLE. This analysis holds in the low-dose regime and under the stylized diagonal acquisition model.

Motivated by these difficulties, practitioners often resort to WLS, which yields a quadratic objective by simplifying the heteroscedastic model. Here we consider the simplest instance of WLS, corresponding to  the \emph{homoscedastic} Gaussian model
\begin{equation}\label{eq:homo_gauss_surr}
y_j \sim \mathcal N\!\big(s(\Am\xb^\star)_j,\; 1\big),
\qquad j=1,\ldots,m.
\end{equation}
Although this model is not distributionally consistent with \eqref{eq:obs_general}, because the variance does not depend on $s(\Am \xb)_j$, we show that it can still yield good MSE performance in the low-dose regime.
More precisely, we study the MAP estimator with  Tikhonov regularization:
\begin{equation} \label{eq:OG_short}
\begin{aligned}
\hat \xb_{\mathrm{Tik,G}}\; :=\; \arg\min_{\xb\in \mathcal X_+} \left\{ \sum_{j=1}^d \frac12\big(y_j-s (\Am  \xb)_j\big)^2 + \frac{\tau}{2} \|\xb\|_2^2 \right\}.
\end{aligned}
\end{equation} 
In the diagonal model \eqref{eq:diag_obs}, this estimator has a closed-form:
\begin{equation}\label{eq:OG_closed_short}
\hat x_{\mathrm{Tik,G},j}(y_j)
=
\begin{cases}
\displaystyle \frac{s a_j}{(s a_j)^2+\tau}\,y_j& j\leq d,\\
\displaystyle 0, & j >d.
\end{cases} \end{equation}
The next proposition characterizes its low-dose MSE  and compares it to that of the unregularized Poisson MLE. 
 
\begin{proposition}\label{prop:per-mode-mse-ogmle}
The per-mode MSE ratio relative to the Poisson MLE satisfies
\begin{equation}\label{eq:OG_ratio_short}
\frac{\mathbb E(\hat x_{\mathrm{Tik,G},j}-x^\star_j)^2}{\mathbb E(\hat x_{\mathrm{MLE,P},j}-x^\star_j)^2}
=
\Big(\frac{1}{1+\gamma_j}\Big)^2
+
\Big(\frac{\gamma_j}{1+\gamma_j}\Big)^2\mu_j,
\end{equation}
where $\gamma_j$ is defined in \eqref{eq:gamma_short}.
\end{proposition}
A similar expression for the \emph{global} MSE ratio (under a uniform low-dose condition over resolved modes) is provided in Supplementary Materials~\ref{subsubsec:global_mse}.

Equation~\eqref{eq:OG_ratio_short}  shows that, in the low-dose regime, i.e. $\mu_j = s a_j x_j^\star \to 0$, the homoscedastic estimator reduces the per-mode MSE of the unregularized Poisson MLE by the explicit factor $(1+\gamma_j)^{-2}<1$ in leading order. Interestingly, this leading-order constant is smaller than the corresponding one for Poisson MAP in \eqref{eq:tik_ratio_short}, meaning that homoscedastic MAP shrinks the per-mode MSE more than Poisson MAP. This provides a positive answer to the question we asked in the Introduction. Indeed, at least in this stylized diagonal setting, a simple Gaussian objective can improve MSE in low-count regimes without requiring a Poisson-specific solver; moreover, since the model is a constant-variance Gaussian, it results in a \textit{linear} estimator.

\section{Experiments}\label{sec:experiments}
We now explore numerically whether these conclusions extend to a more general problem of 2D parallel-beam CT.

\subsection{Experimental setting}
We reconstruct a ground-truth image $\xb^\star \in \Xc_+$ supported on the inscribed circle $\mathcal{C}$ (see Fig.~\ref{fig:qualitative}), with $n=256\times 256$. 
We simulate 2D parallel-beam CT with 180 equispaced angles in $[0,\pi)$ using ASTRA-cuda \cite{VANAARLE201535}.
The noisy data are generated by $\yb \sim \mathcal{P}(s\ \Ab\xb^\star)$. 
The scale $s$ is chosen so that the average expected count per detector bin equals a target value $c$.

\begin{figure*}[h!]
\centering
\begin{tikzpicture}[
  every node/.style={inner sep=0pt},
  title/.style={font=\footnotesize, inner sep=1pt}
]

\def\W{0.18\textwidth}   
\def\H{0.18\textwidth}  
\def\X{2.2mm}           
\def\Y{4.0mm}

\node (r2c1)  {\includegraphics[width=\W,height=\H]{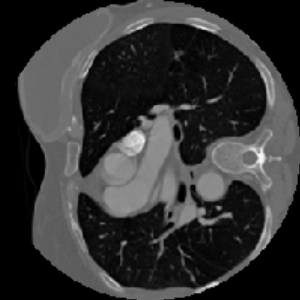}};
\node (r2c2) [right=\X of r2c1] {\includegraphics[width=\W,height=\H]{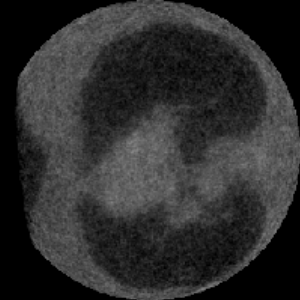}};
\node (r2c3) [right=\X of r2c2] {\includegraphics[width=\W,height=\H]{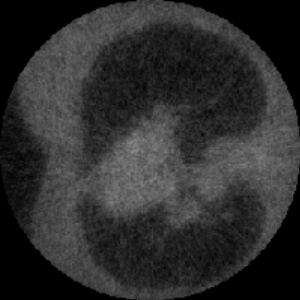}};
\node (r2c4) [right=\X of r2c3] {\includegraphics[width=\W,height=\H]{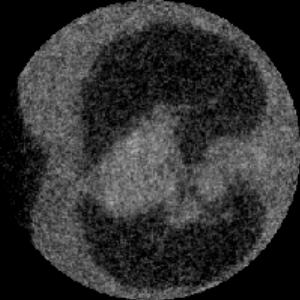}};
\node (r2c5) [right=\X of r2c4] {\includegraphics[width=\W,height=\H]{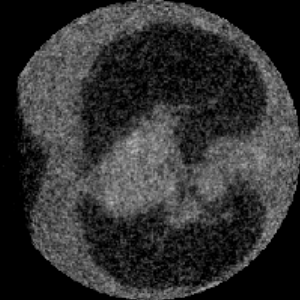}};

=======================
Row 3 (no titles)
  Replace filenames with your 5 images
=======================
\node (r3c1) [below=\Y of r2c1] {\includegraphics[width=\W,height=\H]{csv_files_for_plots/lodo/gt.pdf}};
\node (r3c2) [right=\X of r3c1] {\includegraphics[width=\W,height=\H]{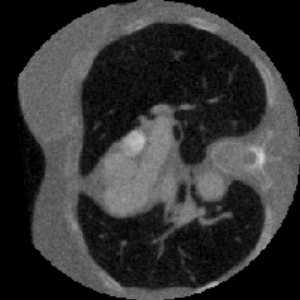}};
\node (r3c3) [right=\X of r3c2] {\includegraphics[width=\W,height=\H]{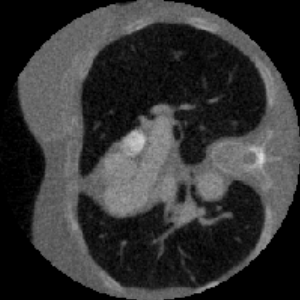}};
\node (r3c4) [right=\X of r3c3] {\includegraphics[width=\W,height=\H]{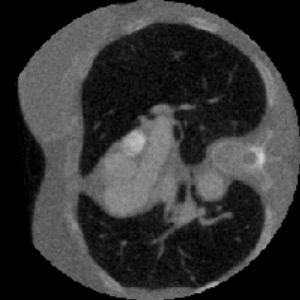}};
\node (r3c5) [right=\X of r3c4] {\includegraphics[width=\W,height=\H]{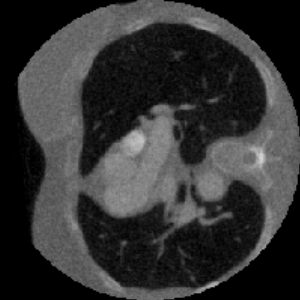}};

\node[title, rotate=90] at (-2.2,-0) {Average expected} ;
\node[title, rotate=90] at (-1.9,-0) {number of counts $c$=10} ;
\node[title, rotate=90] at (-2.2,-3.6) {Average expected} ;
\node[title, rotate=90] at (-1.9,-3.6) {number of counts $c$=1000} ;
\node[title, above=0.5mm of r2c1] {Ground-truth};
\node[title, above=0.5mm of r2c2] {Regularized HG MAP};
\node[title, above=0.5mm of r2c3] {Homoscedastic LS};
\node[title, above=0.5mm of r2c4] {PWLS (oracle)};
\node[title, above=0.5mm of r2c5] {Poisson MAP};

\end{tikzpicture}
\caption{Qualitative reconstructions for a representative CT slice from the LoDoPaB dataset at two count levels.
From left to right: ground truth and reconstructions obtained with regularized HG MAP, homoscedastic LS, PWLS (oracle weights), and Poisson MAP.}
\label{fig:qualitative}
\end{figure*}

We compare three reconstruction objectives: (1) \emph{Poisson MAP} \eqref{eq:pois_tikh} which we optimize using the one-step-late MAP-EM update~\cite{green1990bayesian}; (2) \emph{regularized HG MAP} \eqref{eq:HG_short} with an additional Tikhonov prior $\tfrac{\tau}{2}\|\xb\|^2$, optimized using L-BFGS-B with box constraints implementing $\xb\in\mathcal{C}$; and (3) \emph{penalized weighted least squares (PWLS)}, optimized using L-BFGS-B with box constraints implementing $\xb\in\mathcal{C}$. PWLS solves 
\begin{equation}\label{eq:pwls_exp}
    \arg\min_{\xb\in\mathcal{C}} \frac12\sum_{j\in\mathcal{V}} w_j\big(\mu_j(\xb)-y_j\big)^2 + \frac{\tau}{2}\|\xb\|_2^2,
\end{equation}
where $\mub(\xb)=s\Ab\xb$. We use three settings for the weights: \emph{PWLS (oracle)} with $w_j = (\mu_j(\xb^{\star}) + \varepsilon)^{-1}$, \emph{PWLS (plug-in)} with $w_j = (y_j+\varepsilon)^{-1}$, and \emph{PWLS (plug-in-FBP)} with $w_j = (\mu_j(\hat \xb_{\mathrm{FBP}}) + \varepsilon)^{-1}$. We also include \emph{homoscedastic LS} as the special case $w_j \equiv 1$. The oracle weights use the true variance and serve as the best-case scenario for fixed-weight PWLS.
We exclude the standard Poisson MLE which performs poorly at convergence. 

We run each solver to numerical convergence---avoiding early stopping as an implicit regularizer---and initialize all methods with a uniform, count-matched image (positive constant on $\mathcal{C}$, zero outside). For each method and dose level $s$, we choose the Tikhonov weight $\tau$ by minimizing tuning-set MSE after running the solver to convergence.  In the Supplementary Materials~\ref{app:u_shaped}, we report U-shaped MSE versus regularization strength $\tau$ curves.

Finally, reconstruction quality is measured by MSE on the circular field of view (FOV). We report the mean of this metric over independent test instances, as specified below for our main benchmark\footnote{Code is available at GitHub: \href{https://github.com/AlexandraSpitzer/tomo-reconstruction-gaussian-poisson}{https://github.com/AlexandraSpitzer/tomo-reconstruction-gaussian-poisson}.}

\def\ps{0.45}
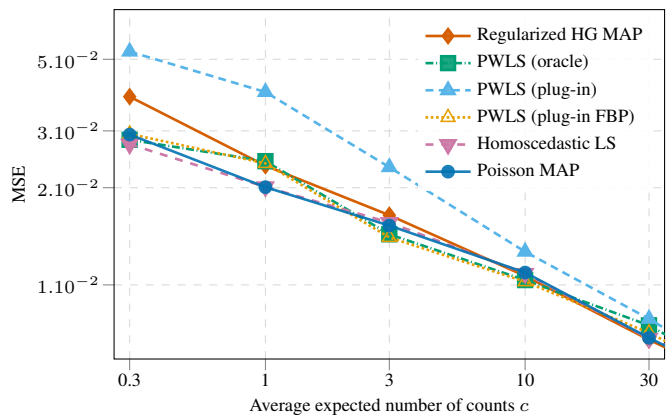
\begin{figure}[h!]
\centering
    \centering
    \begin{tikzpicture}

    \pgfplotsset{
      every axis/.append style={
        tick label style={font=\scriptsize},
        label style={font=\scriptsize},
        xlabel style={yshift=2pt},
        ylabel style={yshift=-1pt},
        ticklabel style={inner sep=1pt},
        enlarge x limits=0.03,
        enlarge y limits=0.08,
        clip mode=individual,
      }
    }

    \begin{groupplot}[
      group style={
        group size=2 by 1,
        horizontal sep=0.85cm,
        vertical sep=0.42cm
      },
      every axis plot/.append style={mark options={solid} 
      },
      width=0.5\textwidth,
      height=0.35\textwidth,
      xmode=log,
      log basis x=10,
      grid=both,
      major grid style={dashed,gray!30},
      minor grid style={dotted,gray!10},
      minor tick num=9,
      xtick={0.3,1,3,10,30,100,300,1000},
      xticklabels={0.3,1,3,10,30,100,300,1000},
      xtick pos=lower,
      ytick pos=left,
      tick align=outside,
      xlabel near ticks,
      ylabel near ticks,
      xlabel={\scriptsize Average expected number of counts $c$},
      ymode=log,
      ymin=7*1e-3,
      ymax=6e-2,
      ytick={0.01,0.02,0.03,0.05},
      yticklabels={$1.10^{-2}$,$2.10^{-2}$,$3.10^{-2}$,$5.10^{-2}$},
      ytickten={-3,-2,-1},
      legend cell align={left},
      legend style={
      at={(axis description cs:0.98,0.98)},
      anchor=north east,
      font=\scriptsize,
      draw=none,
      fill=white,
      fill opacity=0.9,
      text opacity=1,
      inner sep=2pt,
    }
    ]

\nextgroupplot[
  ylabel={\scriptsize MSE},
  title style={
    at={(rel axis cs:0.98,0.90)},
    anchor=north east,
    font=\normalsize,
    inner sep=1pt,
  },
]

\addplot+[OIvermillion, solid, line width=1.0pt,
          mark=diamond*, mark size=2.6pt,
          mark options={solid, draw=OIvermillion, fill=OIvermillion}]
  table[x=target_c,y=HG, col sep=comma]{csv_files_for_plots/lodo/mse_vs_tc-lodo.csv};
\addlegendentry{Regularized HG MAP}

\addplot+[OIgreen, densely dashdotted, line width=1.0pt,
          mark=square*, mark size=2.6pt,
          mark options={solid, draw=OIgreen, fill=OIgreen}]
  table[x=target_c,y={WLS (oracle)}, col sep=comma]{csv_files_for_plots/lodo/mse_vs_tc-lodo.csv};
\addlegendentry{PWLS (oracle)}

\addplot+[OIsky, densely dashed, line width=1.0pt,
          mark=triangle*, mark size=2.8pt,
          mark options={solid, draw=OIsky, fill=OIsky}]
  table[x=target_c,y={WLS (plug-in)}, col sep=comma]{csv_files_for_plots/lodo/mse_vs_tc-lodo.csv};
\addlegendentry{PWLS (plug-in)}

\addplot+[OIorange, densely dotted, line width=1.0pt,
          mark=triangle, mark size=3.0pt,
          mark options={solid, draw=OIorange, fill=white, line width=0.6pt}]
  table[x=target_c,y={WLS (plug-in FBP)}, col sep=comma]{csv_files_for_plots/lodo/mse_vs_tc-lodo.csv};
\addlegendentry{PWLS (plug-in FBP)}

\addplot+[OIpurple, dashed, line width=1.0pt,
          mark=triangle*, mark size=3.2pt,
          mark options={solid, rotate=180, draw=OIpurple, fill=OIpurple}]
  table[x=target_c,y=Homoscedastic, col sep=comma]{csv_files_for_plots/lodo/mse_vs_tc-lodo.csv};
\addlegendentry{Homoscedastic LS}

\addplot+[OIblue, solid, line width=1.0pt,
          mark=*, mark size=2.1pt,
            opacity=0.9,
          mark options={solid, draw=OIblue, fill=OIblue}]
  table[x=target_c,y=Poisson, col sep=comma]{csv_files_for_plots/lodo/mse_vs_tc-lodo.csv};
\addlegendentry{Poisson MAP}
\end{groupplot}
\end{tikzpicture}
\caption{MSE as a function of the average expected number of counts per detector bin. For HG and all PWLS variants, the stabilization floor $\varepsilon$ was selected from $\{0.1,\,0.5,\,1.0\}$ by minimizing the tuning-set MSE at the lowest count level; the chosen $\varepsilon$ was then fixed and used for all average expected number of counts.}
\label{fig:lodopab}
\end{figure}

\subsection{LoDoPaB-CT slices}\label{sec:lodo_exp}
We start with the LoDoPaB-CT dataset \cite{Leuschner2021}, which provides clinical CT slices with corresponding ground-truth reconstructions. This allows us to report performance statistics across many different images. 
We randomly selected 60 patients. Ten disjoint slices were used for tuning and fifty for testing. The slices were downsampled using bicubic interpolation and rescaled to $[0, 1]$ using per-slice min–max normalization. 
For each slice and dose $s$, we form the expected sinogram $\mub=s\Ab\xb^\star$, and draw a single Poisson realization.

Fig.~\ref{fig:qualitative} shows reconstructions at two count levels.
For a relatively small number of counts (mean expected count by detector bin is $10$), the reconstruction obtained with homoscedastic LS or regularized HG MAP is better than Poisson MAP. As shown in Fig.~\ref{fig:lodopab}, the homoscedastic LS and the Poisson MAP give almost the same MSE, despite visual differences. 
At a larger expected number of counts (mean expected count by detector bin is $1000$), anatomical structure becomes more visually coherent for all methods, and there is little or no perceptual difference.

Fig.~\ref{fig:lodopab} reports MSE versus the average expected number of counts. All methods improve monotonically with higher count levels, and the performance gaps vanish in the moderate-to-high count regime. 
In the low count regime, \emph{Poisson MAP is not consistently better in MSE}. Homoscedastic LS (Section~\ref{sec:OG}) as well as regularized HG MAP and PWLS, all match closely the MSE of Poisson MAP, and sometimes slightly outperform it.
This is consistent with the analysis: low-count regularization and implicit shrinkage can be more important for MSE than exact likelihood matching.

In the Supplementary Materials (Section~\ref{app:MSE-shepp-logan}) we additionally report results with the fixed ground truth (the Shepp-Logan phantom) where MSE differences only come from noise realizations and different reconstruction methods. Findings align with those on  LoDoPaB.

\section{Conclusion}
Our analysis shows a counterintuitive result: when reconstructing from Poisson-distributed measurements, the choice between Poisson and Gaussian likelihood objectives matters less for MSE than proper regularization---even at very low dose where one would expect Gaussian approximations to fail. The theoretical analysis, while conducted in a stylized setting (a diagonal forward operator), still clarifies the underlying mechanism: in ill-posed inverse problems, the forward operator attenuates fine-scale modes, pushing them into an effective low-dose regime regardless of the overall photon or electron count. In this regime, rare single-count events can produce large variance spikes that dominate MSE, and regularization acts primarily by damping these modes.
Computed tomography experiments confirm that this behavior extends beyond the diagonal model: across all tested count levels, simple quadratic objectives such as ordinary least squares match or slightly exceed Poisson MAP in MSE. 

While the experimental results align with the theory, we stress that our analysis is intentionally stylized; extending it beyond the diagonal setting considered here would require much greater effort. Moreover, our conclusions are limited to classical likelihood-based estimators and simple Tikhonov-type regularization; deep learning methods based on strong spatial priors may not follow the same conclusions as the algorithms studied in this paper.

\section{Acknowledgments}
The authors would like to thank Jeremy Cohen, Voichita Maxim and Thibaut Modrzyk for valuable discussions regarding Poisson inverse problems.
V.D. is supported by the Agence National de la Recherche (ANR) and the Ministère de l'Enseignement Supérieur et de la Recherche.
Calculations were performed at sciCORE (\href{https://scicore.unibas.ch/}{https://scicore.unibas.ch/}).

\bibliographystyle{IEEEbib}
\bibliography{refs}

\clearpage
\newpage
\onecolumn
\appendix
\begin{center}
    \section*{Supplementary Materials for ``Gaussian Surrogates for Poisson Imaging: \\ Some Theoretical and Empirical Results''}
\end{center}

\section{Proofs of Section~\ref{SEC:METHODS}}

\subsection{Proof of Proposition~\ref{prop:per-mode-mse-pmle}}
\subsubsection{Per-mode Poisson MAP expression}
Recall the standard bias-variance decomposition over modes:
\[
\mathbb E\big\| \hat{\xb}_{{\mathrm{Tik,P}}}(\yb)-\xb^\star \big\|^2
=
\sum_{i=1}^d \mathbb E \big(\hat x_{\mathrm{Tik,P},i}(y_i)-x^\star_i\big)^2
\;+\;
\sum_{i>d}x_i^{\star2},
\]
where, for each $i$,
\[
\mathbb E \big(\hat x_{\mathrm{Tik,P},i}(y_i)-x^\star_i\big)^2
=
\sum_{k=0}^\infty \big( \hat{x}_{\mathrm{Tik,P},i}(k)-x^\star_i\big)^2\Pr(y_i=k),
\]
and
\[
\Pr(y_i=k)=e^{-\mu_i}\frac{\mu_i^k}{k!},
\qquad
\mu_i:=s a_i x^\star_i.
\]
When $\mu_i\to 0$ (low dose for mode $i$), we have
\[
\begin{aligned}
&\Pr(y_i=0)=e^{-\mu_i}=1-\mu_i+O(\mu_i^2), \qquad \Pr(y_i=1)=\mu_i e^{-\mu_i}=\mu_i+O(\mu_i^2),
\end{aligned}
\]
and hence $\Pr(y_i\ge 2)=O(\mu_i^2)$.
Set
\[
h_i := \hat x_{\mathrm{Tik,P},i}(1)=\frac{2}{s a_i+\sqrt{(s a_i)^2+4\tau}}.
\]
Splitting the expectation into $k=0$, $k=1$ and $k\ge 2$ gives
\[
\begin{aligned}
\mathbb E(\hat x_{\mathrm{Tik,P},i}(y_i)-x^\star_i)^2 =
x_i^{\star 2}\,\Pr(y_i=0)
+
(h_i-x^\star_i)^2\,\Pr(y_i=1)
+
R_i,
\end{aligned}
\]
where
\[
R_i := \sum_{k\ge 2} \big(\hat x_{\mathrm{Tik,P},i}(k)-x^\star_i\big)^2\Pr(y_i=k).
\]
We bound $R_i$ explicitly. For all $k\geq 2$, we have
\[
\hat x_{\mathrm{Tik,P},i}(k)=\frac{2k}{s a_i+\sqrt{(s a_i)^2+4\tau k}}\le \frac{k}{s a_i},
\]
and it follows that
\[
(\hat x_{\mathrm{Tik,P},i}(k)-x^\star_i)^2 \le 2\Big(\frac{k}{s a_i}\Big)^2 + 2x_i^{\star 2}.
\]
Therefore,
\[
R_i \le \frac{2}{(s a_i)^2}\,\mathbb E\!\big[y_i^2\mathbf 1_{\{y_i\ge 2\}}\big] \;+\; 2x_i^{\star 2}\,\Pr(y_i\ge 2).
\]
For integers $k\ge 2$, one has $k^2 \le 2k(k-1)$, hence
\[
y_i^2\mathbf 1_{\{y_i\ge 2\}} \le 2y_i(y_i-1).
\]
Taking expectations and using $\mathbb E[y_i(y_i-1)]=\mu_i^2$, it yields
\[
\mathbb E\!\big[y_i^2\mathbf 1_{\{y_i\ge 2\}}\big]\le 2\mu_i^2.
\]
Moreover, since $y_i(y_i-1)\ge 2\,\mathbf 1_{\{y_i\ge 2\}}$,
\[
\Pr(y_i\ge 2)\le \frac{\mathbb E[y_i(y_i-1)]}{2}=\frac{\mu_i^2}{2}.
\]
Consequently,
\[
R_i \le \frac{4\mu_i^2}{(s a_i)^2} + x_i^{\star 2}\mu_i^2.
\]
Using $\Pr(y_i=0)=1-\mu_i+O(\mu_i^2)$ and $\Pr(y_i=1)=\mu_i+O(\mu_i^2)$, we obtain
\[
\begin{aligned}
\mathbb E(\hat x_{\mathrm{Tik,P},i}(y_i)-x^\star_i)^2
&=
x_i^{\star 2}(1-\mu_i+O(\mu_i^2))
+
(h_i-x^\star_i)^{2}(\mu_i+O(\mu_i^2))
+R_i\\
&=
x_i^{\star 2}
\!+\!\mu_i(h_i^2\!-\!2h_i x^\star_i)
+O\!\big(\mu_i^2(h_i^2\!+\!x_i^{\star2})\big)
\!+\!\frac{4\mu_i^2}{(s a_i)^2}
\!+\!x_i^{\star 2}\mu_i^2,
\end{aligned}
\]
where we used $(h_i-x^\star_i)^2\le 2h_i^2+2x_i^{\star 2}$.
In particular,
\[
\begin{aligned}
\mathbb E(\hat x_{\mathrm{Tik,P},i}(y_i)-x^\star_i)^2 =
x_i^{\star 2}
+\mu_i(h_i^2-2h_i x^\star_i)
+O\!\Big(\mu_i^2 h_i^2 + \frac{\mu_i^2}{(s a_i)^2} + x_i^{\star 2}\mu_i^2\Big).
\end{aligned}
\]

\subsubsection{Per-mode MSE ratio}
The unregularized Poisson MLE is $\hat x_{\mathrm{MLE,P},i}=y_i/(s a_i)$, and
\[
\mathbb E(\hat x_{\mathrm{MLE,P},i}-x^\star_i)^2=\frac{x^\star_i}{s a_i}=\frac{\mu_i}{(s a_i)^2}.
\]
Dividing the MSE of the Poisson MAP by $\mu_i/(s a_i)^2$ gives
\[
\begin{aligned}
\frac{\mathbb E(\hat x_{\mathrm{Tik,P},i}(y_i)-x^\star_i)^2}{\mathbb E(\hat x_{\mathrm{MLE,P},i}-x^\star_i)^2}
= 
\frac{(s a_i)^2}{\mu_i}\,x_i^{\star 2}
+
(s a_i)^2 h_i^2
-
2(s a_i)^2 h_i x^\star_i
+
O\!\Big((s a_i)^2\mu_i h_i^2\Big)
+
O(\mu_i)
 +
O\!\Big((s a_i)^2 x_i^{\star 2}\mu_i\Big).
\end{aligned}
\]
Using $\frac{(s a_i)^2}{\mu_i}x_i^{\star 2}=\mu_i$ and $(s a_i)^2 h_i x^\star_i=\mu_i(s a_i h_i)$, this can be rewritten as
\[
\begin{aligned} \frac{\mathbb E(\hat x_{\mathrm{Tik,P},i}(y_i)-x^\star_i)^2}{\mathbb E(\hat x_{\mathrm{MLE,P},i}-x^\star_i)^2}
=
\mu_i
+
(s a_i)^2 h_i^2
-
2\mu_i(s a_i h_i) +
O\!\Big(\mu_i (s a_i)^2 h_i^2\Big)
+
O(\mu_i)
+
O(\mu_i^3).
\end{aligned}
\]
Since $0< s a_i h_i\le 1$ and $0<(s a_i)^2 h_i^2\le 1$ (because $h_i\le 1/(s a_i)$), the last expression simplifies to
\[
\frac{\mathbb E(\hat x_{\mathrm{Tik,P},i}(y_i)-x^\star_i)^2}{\mathbb E(\hat x_{\mathrm{MLE,P},i}-x^\star_i)^2}
=
(s a_i)^2 h_i^2 + O(\mu_i),
\qquad \mu_i\to 0.
\]
This statement alone does not identify the leading order unless one compares the sizes of $(s a_i)^2 h_i^2$ and $\mu_i$.
More precisely, the expansion
\[
\frac{\mathbb E(\hat x_{\mathrm{Tik,P},i}(y_i)-x^\star_i)^2}{\mathbb E(\hat x_{\mathrm{MLE,P},i}-x^\star_i)^2}
=
(s a_i)^2 h_i^2 + O(\mu_i)
\]
implies
\[
\frac{\mathbb E(\hat x_{\mathrm{Tik,P},i}(y_i)-x^\star_i)^2}{\mathbb E(\hat x_{\mathrm{MLE,P},i}-x^\star_i)^2}
\sim (s a_i)^2 h_i^2,
\]
if and only if
\[
\mu_i = o\!\big((s a_i)^2 h_i^2\big),
\qquad \mu_i\to 0.
\]

We can express $(s a_i)^2 h_i^2$ explicitly as
\[
(s a_i)^2 h_i^2
=
\Big(\frac{2s a_i}{s a_i+\sqrt{(s a_i)^2+4\tau}}\Big)^2.
\]
Define
\[
\gamma_i:=\frac{\tau}{(s a_i)^2}.
\]
The quantity $\gamma_i$ can be viewed as an effective regularization level: the larger $\gamma_i$, the stronger the shrinkage induced by the penalty relative to the Poisson sensitivity of the forward map.
Indeed,
\[
(s a_i)^2 h_i^2=\Big(\frac{2}{1+\sqrt{1+4\gamma_i}}\Big)^2,
\]
so that, as $\mu_i=s a_i x^\star_i\to 0$,
\[
\frac{\mathbb E(\hat x_{\mathrm{Tik,P},i}(y_i)-x^\star_i)^2}{\mathbb E(\hat x_{\mathrm{MLE,P},i}-x^\star_i)^2}
=
\Big(\frac{2}{1+\sqrt{1+4\gamma_i}}\Big)^2 + O(\mu_i).
\]
\begin{remark}
There are three regimes:
\begin{itemize}
\item If $\gamma_i\ll 1$ (weak regularization), then
\[
\Big(\frac{2}{1+\sqrt{1+4\gamma_i}}\Big)^2 = 1+O(\gamma_i),
\]
and the regularized estimator behaves essentially like the Poisson MLE.
\item If $\gamma_i\approx 1$ (balanced regime), then
\[
\Big(\frac{2}{1+\sqrt{5}}\Big)^2=\frac{3-\sqrt5}{2}.
\]
\item If $\gamma_i\gg 1$ (strong regularization), then
\[
\Big(\frac{2}{1+\sqrt{1+4\gamma_i}}\Big)^2 \sim \frac{1}{\gamma_i}=\frac{(s a_i)^2}{\tau},
\qquad \gamma_i\to\infty.
\]
In particular, if $s a_i\to 0$ with $\tau$ fixed (e.g.\ when $s\to 0$ or when $a_i\to 0$ at fine scales), then $\gamma_i\to\infty$ and
\[
(s a_i)^2 h_i^2
=
\frac{(s a_i)^2}{\tau}
+
O\!\Big(\frac{(s a_i)^3}{\tau^{3/2}}\Big),
\qquad s a_i\to 0,
\]
so that the per-mode MSE ratio satisfies
\[
\frac{\mathbb E(\hat x_{\mathrm{Tik,P},i}(y_i)-x^\star_i)^2}{\mathbb E(\hat x_{\mathrm{MLE,P},i}-x^\star_i)^2}
=
\frac{(s a_i)^2}{\tau}
+
O\!\Big(\frac{(s a_i)^3}{\tau^{3/2}}\Big)
+
O(\mu_i).
\]
\end{itemize}
Essentially, $\gamma_i\approx 1$ identifies the cutoff scale where regularization starts to dominate.
\end{remark}

\subsubsection{Global MSE ratio}\label{sec:global-mse-tik}

We now discuss what the mode-by-mode analysis implies for the global MSE.
Recall
\[
\begin{aligned} 
\mathbb E\big\|\hat \xb_{\mathrm{Tik,P}}(\yb)-\xb^\star\big\|^2
&=
\sum_{i=1}^d \mathbb E\big(\hat x_{\mathrm{Tik,P},i}(y_i)-x^\star_i\big)^2
+
\sum_{i>d} x_i^{\star 2},\\
\mathbb E\big\|\hat \xb_{\mathrm{MLE,P}}(\yb)-\xb^\star\big\|^2
&=
\sum_{i=1}^d \mathbb E\big(\hat x_{\mathrm{MLE,P},i}(y_i)-x^\star_i\big)^2
+
\sum_{i>d} x_i^{\star2}.
\end{aligned}
\]
Define
\[
V_d := \sum_{i=1}^d \frac{x^\star_i}{sa_i}, \qquad B_d:=\sum_{i=1}^d x_i^{\star 2}, \qquad T_d:= \sum_{i>d} x_i^{\star 2},
\]
so that
\[
\mathbb E \|\hat x_{\mathrm{MLE,P}}(\yb) -x^\star\|^2 = V_d + T_d.
\]
For each $i$, set
\[
h_i:=\hat x_{\mathrm{Tik,P},i}(1)=\frac{2}{s a_i+\sqrt{(s a_i)^2+4\tau}},\]
\[
\gamma_i:=\frac{\tau}{(s a_i)^2},
\]
so that
\[
(s a_i)^2 h_i^2=\Big(\frac{2}{1+\sqrt{1+4\gamma_i}}\Big)^2
=:c_P(\gamma_i)^2.
\]
From the per-mode low-dose expansion,

\[
\frac{\mathbb E(\hat x_{\mathrm{Tik,P},i}(y_i)-x^\star_i)^2}{\mathbb E(\hat x_{\mathrm{MLE,P},i}-x^\star_i)^2}
=
(s a_i)^2 h_i^2 + O(\mu_i)\]
as $\mu_i\to 0$,
and since $\mathbb E(\hat x_{\mathrm{MLE,P},i}-x^\star_i)^2 = x^\star_i/(s a_i)$, this is equivalent to

\[
\mathbb E(\hat x_{\mathrm{Tik,P},i}(y_i)-x^\star_i)^2
=
c_P(\gamma_i)^2\,\frac{x^\star_i}{s a_i}
+
r_i,
\qquad
r_i
=
O\!\Big(\mu_i\,\frac{x^\star_i}{s a_i}\Big),
\]
as $\mu_i\to 0$, where $c_P(\gamma_i)^2\,\frac{x^\star_i}{s a_i}$ should be read as the leading term \emph{relative to} the Poisson MLE scale $x^\star_i/(s a_i)$. Using $\mu_i\,\frac{x^\star_i}{s a_i}=x_i^{\star 2}$, the remainder can be written as
\[
r_i = O(x_i^{\star 2}).
\]
More explicitly, from the MSE expansion
\[
\begin{aligned}
\mathbb E(\hat x_{\mathrm{Tik,P},i}(y_i)-x^\star_i)^2
=
x_i^{\star 2}
+\mu_i\big(h_i^2-2h_i x^\star_i\big) +
O\!\Big(\mu_i^2 h_i^2 + \frac{\mu_i^2}{(s a_i)^2} + x_i^{\star 2}\mu_i^2\Big)
\end{aligned}
\]
as $\mu_i\to 0$, we obtain

\[
\begin{aligned}
\mathbb E(\hat x_{\mathrm{Tik,P},i}(y_i)-x^\star_i)^2
= c(\gamma_i)^2\,\frac{x^\star_i}{s a_i}
+
\bigl(1-2\,s a_i h_i\bigr)x_i^{\star 2}
+
O(x_i^{\star 2})
+
O\bigl((s a_i)^2 x_i^{\star 4}\bigr).
\end{aligned} 
\]
When passing to the global MSE, we sum these remainders over $i\le d$; under the uniform low-dose condition
$\max_{i\le d(s)}\mu_i\to 0$, the constants implicit in the $O(\cdot)$ terms may be chosen uniformly over $i\le d(s)$, so that the summed remainders are
$O(B_d)$ and $O\!\big(\sum_{i=1}^d (s a_i)^2 x_i^{\star 4}\big)$.

Summing over $i\le d$ then yields the global expansion
\[
\begin{aligned}
\mathbb E\big\|\hat \xb_{\mathrm{Tik,P}}(\yb)-\xb^\star\big\|^2
= &
\sum_{i=1}^d c_P(\gamma_i)^2\,\frac{x^\star_i}{s a_i}
+
\sum_{i=1}^d \bigl(1-2\,s a_i h_i\bigr)x_i^{\star 2}   +
T_d
+
O(B_d)
+
O\!\Big(\sum_{i=1}^d (s a_i)^2 x_i^{\star 4}\Big).
\end{aligned}
\]

Dividing by the MSE relative to Poisson MLE gives
\[
\begin{aligned}
\frac{\mathbb E\|\hat \xb_{\mathrm{Tik,P}}(\yb)-\xb^\star\|^2}{\mathbb E\|\hat \xb_{\mathrm{MLE,P}}(\yb)-\xb^\star\|^2}
=  \frac{\sum_{i=1}^d \left( c_P(\gamma_i)^2\,\frac{x^\star_i}{s a_i}
+
\bigl(1-2\,s a_i h_i\bigr)x_i^{\star 2}\right)+ T_d
+
O(B_d)
+
O\!\Big(\sum_{i=1}^d (s a_i)^2 x_i^{\star 4}\Big)
}{V_d+T_d}.
\end{aligned}
\]

In particular, assume that $d=d(s)$ is chosen so that 
\[
\frac{B_{d(s)}}{V_{d(s)}+T_{d(s)}}\to 0,
\qquad
\frac{T_{d(s)}}{V_{d(s)}+T_{d(s)}}\to 0, \qquad \frac{\sum_{i=1}^{d(s)} (s a_i)^2 x_i^{\star 4}}{V_{d(s)}+T_{d(s)}}\to 0
\]
as $s\to 0$, and (uniform per-mode low-dose regime)
\[
\max_{1\le i\le d(s)} \mu_i
\to 0.
\]
Under these assumptions, the error terms in the numerator are negligible compared to $V_{d(s)}+T_{d(s)}$, and therefore the asymptotic behavior of the global ratio is driven by the leading term
\[
\frac{\sum_{i=1}^{d(s)} c_P(\gamma_i)^2\,\frac{x^\star_i}{s a_i}}{V_{d(s)}+T_{d(s)}}.
\]
Since $T_{d(s)}/(V_{d(s)}+T_{d(s)})\to 0$, it suffices to compare to $V_{d(s)}$.
Define the $V_d$-weighted average
\[
\bar c_{P,d}^{\,2}
:=
\frac{1}{V_d}\sum_{i=1}^d c_P(\gamma_i)^2\,\frac{x^\star_i}{s a_i}.
\]
Then the global ratio admits the expression
\[
\frac{\mathbb E\|\hat \xb_{\mathrm{Tik,P}}(\yb)-\xb^\star\|^2}{\mathbb E\|\hat \xb_{\mathrm{MLE,P}}(\yb)-\xb^\star\|^2}
=
\bar c_{P,d(s)}^{\,2}
+
o(1),
\qquad s\to 0.
\]

\subsection{Proof of Proposition~\ref{prop:per-mode-mse-ogmle}}
\subsubsection{Per-mode homoscedastic Gaussian MAP expression}
As we did for the proof of Proposition \ref{prop:per-mode-mse-pmle}, we begin by writing down the bias-variance decomposition over modes:
\[
\mathbb E\big\| \hat{\xb}_{{\mathrm{Tik,G}}}(\yb)-\xb^\star \big\|^2
=
\sum_{i=1}^d \mathbb E \big(\hat x_{\mathrm{Tik,G},i}(y_i)-x^\star_i\big)^2
\;+\;
\sum_{i>d}x_i^{\star 2}.
\]
Set
\[
h_i := \hat x_{\mathrm{Tik,G},i}(1)=\frac{s a_i}{(s a_i)^2+\tau},
\]
so that 
\[
\hat x_{\mathrm{Tik,G},i}(y_i)=h_i y_i.
\]
Then
\[
\hat x_{\mathrm{Tik,G},i}(y_i)-x^\star_i
=
\bigl(s a_i h_i-1\bigr)x^\star_i + h_i (y_i-\mu_i),
\]
and hence 
\[
\mathbb E \big(\hat x_{\mathrm{Tik,G},i}(y_i)-x^\star_i\big)^2
=
\bigl(1-s a_i h_i\bigr)^2 x_i^{\star 2} + h_i^2\,\mu_i.
\]
Using $x^\star_i=\mu_i/(s a_i)$, we obtain
\[
\mathbb E \big(\hat x_{\mathrm{Tik,G},i}(y_i)-x^\star_i\big)^2
=
h_i^2\,\mu_i
+
\bigl(1-s a_i h_i\bigr)^2\frac{\mu_i^2}{(s a_i)^2}.
\]

\subsubsection{Per-mode MSE ratio}
The unregularized Poisson MLE is $\hat x_{\mathrm{MLE,P},i}=y_i/(s a_i)$, and
\[
\mathbb E(\hat x_{\mathrm{MLE,P},i}-x^\star_i)^2=\frac{x^\star_i}{s a_i}=\frac{\mu_i}{(s a_i)^2}.
\]
Dividing by $\mu_i/(s a_i)^2$ gives
\[
\frac{\mathbb E(\hat x_{\mathrm{Tik,G},i}(y_i)-x^\star_i)^2}{\mathbb E(\hat x_{\mathrm{MLE,P},i}-x^\star_i)^2}
=
(s a_i)^2 h_i^2
+
\bigl(1-s a_i h_i\bigr)^2\mu_i.
\]
We can express $(s a_i)^2 h_i^2$ explicitly as
\[
(s a_i)^2 h_i^2
=
\Big(\frac{(s a_i)^2}{(s a_i)^2+\tau}\Big)^2.
\]
Define
\[
\gamma_i:=\frac{\tau}{(s a_i)^2}.
\]
We have,
\[
(s a_i)^2 h_i^2=\Big(\frac{1}{1+\gamma_i}\Big)^2,
\]
and
\[
1-s a_i h_i=\frac{\tau}{(s a_i)^2+\tau}=\frac{\gamma_i}{1+\gamma_i},
\]
so that
\[
\frac{\mathbb E(\hat x_{\mathrm{Tik,G},i}(y_i)-x^\star_i)^2}{\mathbb E(\hat x_{\mathrm{MLE,P},i}-x^\star_i)^2}
=
\Big(\frac{1}{1+\gamma_i}\Big)^2
+
\Big(\frac{\gamma_i}{1+\gamma_i}\Big)^2\mu_i.
\]
\begin{remark}
As in Proposition \ref{prop:per-mode-mse-pmle}, there are three regimes:
\begin{itemize}
\item If $\gamma_i\ll 1$ (weak regularization), then
\[
\Big(\frac{1}{1+\gamma_i}\Big)^2 = 1+O(\gamma_i),
\]
and the regularized estimator behaves essentially like the Poisson MLE in the low-dose regime.
\item If $\gamma_i\approx 1$ (balanced regime), then
\[
\Big(\frac{1}{1+1}\Big)^2=\frac14,
\]
and
\[
\frac{\mathbb E(\hat x_{\mathrm{Tik,G},i}(y_i)-x^\star_i)^2}{\mathbb E(\hat x_{\mathrm{MLE,P},i}-x^\star_i)^2}
=
\frac14+\frac14\,\mu_i.
\]
\item If $\gamma_i\gg 1$ (strong regularization), then
\[
\Big(\frac{1}{1+\gamma_i}\Big)^2 \sim \frac{1}{\gamma_i^2}=\frac{(s a_i)^4}{\tau^2},
\qquad \gamma_i\to\infty.
\]
In particular, if $s a_i\to 0$ with $\tau$ fixed (e.g.\ when $s\to 0$ or when $a_i\to 0$ at fine scales), then $\gamma_i\to\infty$ and
\[
(s a_i)^2 h_i^2
=
\frac{(s a_i)^4}{\tau^2}
+
O\!\Big(\frac{(s a_i)^6}{\tau^{3}}\Big),
\qquad s a_i\to 0,
\]
so that the per-mode MSE ratio satisfies
\[
\begin{aligned}
\frac{\mathbb E(\hat x_{\mathrm{Tik,G},i}(y_i)-x^\star_i)^2}{\mathbb E(\hat x_{\mathrm{MLE,P},i}-x^\star_i)^2}
=
\frac{(s a_i)^4}{\tau^2}
+
\mu_i
+
O\!\Big(\frac{(s a_i)^6}{\tau^{3}}\Big)
+
O\!\Big(\mu_i\,\frac{(s a_i)^2}{\tau}\Big),
\end{aligned} 
\]
where the leading term is $\frac{(s a_i)^4}{\tau^2}$ if and only if $\mu_i=o\!\big((s a_i)^4/\tau^2\big)$.
\end{itemize}
Essentially, $\gamma_i\approx 1$ identifies the cutoff scale where regularization starts to dominate.
\end{remark}

\subsubsection{Global MSE ratio}\label{subsubsec:global_mse}

We now discuss what the mode-by-mode analysis implies for the global MSE. As in the proof of Proposition \ref{prop:per-mode-mse-pmle}, 
define
\[
V_d := \sum_{i=1}^d \frac{x^\star_i}{s a_i}, \qquad B_d:=\sum_{i=1}^d x_i^{\star 2}, \qquad T_d:= \sum_{i>d} x_i^{\star 2},
\]
so that
\[
\mathbb E \|\hat \xb_{\mathrm{MLE,P}}(\yb) -\xb^\star\|^2 = V_d + T_d.
\]
For each $i$, set
\[
h_i:=\frac{s a_i}{(s a_i)^2+\tau},
\qquad
\gamma_i:=\frac{\tau}{(s a_i)^2},
\]
so that
\[
(s a_i)^2 h_i^2=\Big(\frac{1}{1+\gamma_i}\Big)^2
=:c_G(\gamma_i)^2.
\]
From the per-mode identity,
\[
\frac{\mathbb E(\hat x_{\mathrm{Tik,G},i}(y_i)-x^\star_i)^2}{\mathbb E(\hat x_{\mathrm{MLE,P},i}-x^\star_i)^2}
=
(s a_i)^2 h_i^2
+
\bigl(1-s a_i h_i\bigr)^2\mu_i,
\]
and since $\mathbb E(\hat x_{\mathrm{MLE,P},i}-x^\star_i)^2 = x^\star_i/(s a_i)$, this is equivalent to
\[
\mathbb E(\hat x_{\mathrm{Tik,G},i}(y_i)-x^\star_i)^2
=
c_G(\gamma_i)^2\,\frac{x^\star_i}{s a_i}
+
\bigl(1-s a_i h_i\bigr)^2 x_i^{\star 2}.
\]
Using $1-s a_i h_i=\gamma_i/(1+\gamma_i)$, we obtain
\[
\mathbb E(\hat x_{\mathrm{Tik,G},i}(y_i)-x^\star_i)^2
=
c_G(\gamma_i)^2\,\frac{x^\star_i}{s a_i}
+
\Big(\frac{\gamma_i}{1+\gamma_i}\Big)^2 x_i^{\star 2}.
\]

Summing over $i\le d$ yields the global identity
\[
\mathbb E\big\|\hat \xb_{\mathrm{Tik,G}}(\yb)-\xb^\star\big\|^2
=
\sum_{i=1}^d c_G(\gamma_i)^2\,\frac{x^\star_i}{s a_i}
+
\sum_{i=1}^d \Big(\frac{\gamma_i}{1+\gamma_i}\Big)^2 x_i^{\star 2}
+
T_d.
\]

Dividing by the MSE relative to Poisson MLE gives
\[
\frac{\mathbb E\|\hat \xb_{\mathrm{Tik,G}}(\yb)-\xb^\star\|^2}{\mathbb E\|\hat \xb_{\mathrm{MLE,P}}(\yb)-\xb^\star\|^2}
=
\frac{
\sum_{i=1}^d c_G(\gamma_i)^2\,\frac{x^\star_i}{s a_i}
+
\sum_{i=1}^d \Big(\frac{\gamma_i}{1+\gamma_i}\Big)^2 x_i^{\star 2}
+
T_d
}{V_d+T_d}.
\]

In particular, assume that $d=d(s)$ is chosen so that 
\[
\frac{B_{d(s)}}{V_{d(s)}+T_{d(s)}}\to 0,
\qquad
\frac{T_{d(s)}}{V_{d(s)}+T_{d(s)}}\to 0,
\qquad s\to 0.
\]
Under these assumptions, the second term in the numerator is negligible compared to $V_{d(s)}+T_{d(s)}$, and therefore the asymptotic behavior of the global ratio is driven by the leading term
\[
\frac{\sum_{i=1}^{d(s)} c_G(\gamma_i)^2\,\frac{x^\star_i}{s a_i}}{V_{d(s)}+T_{d(s)}}.
\]
Since $T_{d(s)}/(V_{d(s)}+T_{d(s)})\to 0$, it suffices to compare to $V_{d(s)}$.
Define the $V_d$-weighted average
\[
\bar c_{G,d}^{\,2}
:=
\frac{1}{V_d}\sum_{i=1}^d c_G(\gamma_i)^2\,\frac{x^\star_i}{s a_i}.
\]
Then the global ratio becomes
\[
\frac{\mathbb E\|\hat \xb_{\mathrm{Tik,G}} (\yb)-\xb^\star\|^2}{\mathbb E\|\hat \xb_{\mathrm{MLE,P}} (\yb)-\xb^\star\|^2}
=
\bar c_{G,d(s)}^{\,2}
+
o(1),
\qquad s\to 0.
\]

\subsection{Heteroscedastic Gaussian Surrogate}\label{sec:HG}
Following the discussion at the beginning of Section~\ref{sec:OG}, we consider the heteroscedastic Gaussian approximation $y_j \approx \mathcal{N}(s(\Am \xb^\star)_j, s(\Am \xb^\star)_j)$. To ensure that the resulting objective is well-defined (and the minimization problem well-posed), we introduce a stabilizing floor $\epsilon>0$ and define the corresponding negative log-likelihood as 
\begin{equation} \label{eq:HG_short}
\begin{aligned}
\mathcal L_{\mathrm{HG}}^{(\epsilon)}(\xb;\yb) :=
\sum_{j=1}^d\left(
\frac12\log\!\big(s (\Am \xb)_j+\epsilon\big)
+
\frac{(\yb_j-s (\Am \xb)_j)^2}{2\big(s (\Am \xb)_j+\epsilon\big)}
\right).
\end{aligned}
\end{equation}
The heteroscedastic Gaussian (HG) MLE is  
\begin{equation}\label{eq:HG-MLE}
    \hat \xb_{\mathrm{HG}} \in \arg\min_{\xb\in \mathcal X_+}\mathcal L_{\mathrm{HG}}^{(\epsilon)}(\xb;\yb).
\end{equation}
In the diagonal model \eqref{eq:diag_obs}, the HG MLE becomes
\begin{equation}\label{eq:HG_closed_short}
\hat x_{\text{HG},j}(y_j)
=
\max\left\{
\frac{1}{s a_j}\left(\frac{-1+\sqrt{1+4(y_j+\epsilon)^2}}{2}-\epsilon\right),
\;0
\right\},
\end{equation}
The next proposition characterizes the low-dose MSE of the HG estimator and compares it to that of the unregularized Poisson MLE. 
\begin{proposition}\label{prop:per-mode-mse-hgmle}
Let $\epsilon>0$ and define
\begin{equation}\label{eq:ceps_short}
c(\epsilon):=
\left(\frac{-1+\sqrt{1+4(1+\epsilon)^2}}{2}-\epsilon\right).
\end{equation}
Then, for all $1\leq j\leq d$ such that $\mu_j\to 0$, we have 
\begin{equation}\label{eq:HG_ratio_short}
\frac{\EE(\hat x_{\text{HG},j}-x_j^\star)^2}{\EE(\hat x_{\mathrm{MLE,P},j}-x_j^\star)^2}
=
c(\epsilon)^2 + O(\mu_j),
\qquad
c(\epsilon)^2\in\Big(\tfrac14,\tfrac{3-\sqrt5}{2}\Big).
\end{equation}
\end{proposition}
\begin{proof}
By an abuse of notation, we write $\hat{\xb}_{\text{HG}}(\yb)$ for the embedding
\[
\bigl(\hat \xb_{\text{HG}}(\yb)\bigr)_i :=
\begin{cases}
\hat x_{\text{HG},i}(y_i), & i\le d,\\[2pt]
0, & i>d.
\end{cases}
\]
We have the standard bias-variance decomposition over modes:
\[
\begin{aligned}
\mathbb E\big\| \hat{\xb}_{\text{HG}}(\yb)-\xb^\star \big\|^2
&=
\sum_{i=1}^d \mathbb E \big(\hat x_{\text{HG},i}(y_i)-x^\star_i\big)^2
\;+\;
\sum_{i>d}x_i^{\star 2}.
\end{aligned}
\]
For each $i$,
\[
\begin{aligned}
\mathbb E \big(\hat x_{\text{HG},i}(y_i)-x^\star_i\big)^2
&=
\sum_{k=0}^\infty \big( \hat{x}_{\text{HG},i}(k)-x^\star_i\big)^2\Pr(y_i=k),\\
\Pr(y_i=k)
&=e^{-\mu_i}\frac{\mu_i^k}{k!},
\qquad
\mu_i:=s a_i x^\star_i.
\end{aligned}
\]

Assume $\mu_i\to 0$. Then
\[
\begin{aligned}
\Pr(y_i=0)
=e^{-\mu_i}=1-\mu_i+O(\mu_i^2),\qquad 
\Pr(y_i=1) =\mu_i e^{-\mu_i}=\mu_i+O(\mu_i^2),
\end{aligned}
\]
and hence $\Pr(y_i\ge 2)=O(\mu_i^2)$.
Set
\[
\begin{aligned}
g_i := \hat x_{\text{HG},i}(1)
&=
\frac{1}{s a_i}\left(
\frac{-1+\sqrt{1+4(1+\epsilon)^2}}{2}-\epsilon
\right).
\end{aligned}
\]
Splitting the expectation into $k=0$, $k=1$ and $k\ge 2$ gives
\[
\begin{aligned}
\mathbb E(\hat x_{\text{HG},i}(y_i)-x^\star_i)^2 =
x_i^{\star 2}\,\Pr(y_i=0)
+
(g_i-x^\star_i)^2\,\Pr(y_i=1)
+
\widetilde R_i,
\end{aligned}
\]
where
\[
\widetilde R_i
:=
\sum_{k\ge 2} \big(\hat x_{\text{HG},i}(k)-x^\star_i\big)^2\Pr(y_i=k).
\]

We bound $\widetilde R_i$ explicitly. For any $k\ge 0$, using \[ 
\sqrt{1+4(k+\epsilon)^2}\le 1+2(k+\epsilon)
\]
we obtain
\[
\frac{-1+\sqrt{1+4(k+\epsilon)^2}}{2}-\epsilon \le k,
\]
and therefore $\hat x_{\text{HG},i}(k)\le k/(s a_i)$.
It follows that
\[
(\hat x_{\text{HG},i}(k)-x^\star_i)^2
\le
2\Big(\frac{k}{s a_i}\Big)^2 + 2x_i^{\star 2}.
\]
Therefore,
\[
\begin{aligned}
\widetilde R_i
&\le
\frac{2}{(s a_i)^2}\,
\mathbb E\!\big[y_i^2\mathbf 1_{\{y_i\ge 2\}}\big]
\;+\;
2x_i^{\star 2}\,\Pr(y_i\ge 2).
\end{aligned}
\]
For integers $k\ge 2$ one has $k^2 \le 2k(k-1)$, hence
\[
y_i^2\mathbf 1_{\{y_i\ge 2\}} \le 2y_i(y_i-1).
\]
Taking expectations and using $\mathbb E[y_i(y_i-1)]=\mu_i^2$ yields
\[
\mathbb E\!\big[y_i^2\mathbf 1_{\{y_i\ge 2\}}\big]\le 2\mu_i^2.
\]
Moreover, since $y_i(y_i-1)\ge 2\,\mathbf 1_{\{y_i\ge 2\}}$,
\[
\Pr(y_i\ge 2)\le \frac{\mathbb E[y_i(y_i-1)]}{2}=\frac{\mu_i^2}{2}.
\]
Consequently,
\[
\widetilde R_i \le \frac{4\mu_i^2}{(s a_i)^2} + x_i^{\star 2}\mu_i^2.
\]

Using $\Pr(y_i=0)=1-\mu_i+O(\mu_i^2)$ and $\Pr(y_i=1)=\mu_i+O(\mu_i^2)$, we obtain
\[
\begin{aligned}
\mathbb E(\hat x_{\text{HG},i}(y_i)-x^\star_i)^2
&=
x_i^{\star 2}(1-\mu_i+O(\mu_i^2))
+
(g_i-x^\star_i)^2(\mu_i+O(\mu_i^2))
+\widetilde R_i\\
&=
x_i^{\star 2}
+\mu_i(g_i^2-2g_i x^\star_i)
+O\!\big(\mu_i^2(g_i^2+x_i^{\star 2})\big)
+\frac{4\mu_i^2}{(s a_i)^2}
+x_i^{\star 2}\mu_i^2,
\end{aligned}
\]
where we used $(g_i-x^\star_i)^2\le 2g_i^2+2x_i^{\star 2}$.
In particular,
\[
\begin{aligned}
\mathbb E(\hat x_{\text{HG},i}(y_i)-x^\star_i)^2
=
x_i^{\star 2}
+\mu_i(g_i^2-2g_i x^\star_i)
+O\!\Big(
\mu_i^2 g_i^2
+
\frac{\mu_i^2}{(s a_i)^2}
+
x_i^{\star 2}\mu_i^2
\Big).
\end{aligned}
\]
The unregularized Poisson MLE is $\hat x_{\mathrm{MLE,P},i}=y_i/(s a_i)$, and
\[
\mathbb E(\hat x_{\mathrm{MLE,P},i}-x^\star_i)^2=\frac{x^\star_i}{s a_i}=\frac{\mu_i}{(s a_i)^2}.
\]
Dividing by $\mu_i/(s a_i)^2$ gives
\[
\begin{aligned}
\frac{\mathbb E(\hat x_{\text{HG},i}(y_i)-x^\star_i)^2}{\mathbb E(\hat x_{\mathrm{MLE,P},i}-x^\star_i)^2}
=
\frac{(s a_i)^2}{\mu_i}\,x_i^{\star 2}
+
(s a_i)^2 g_i^2
-
2(s a_i)^2 g_i x^\star_i +
O\!\Big((s a_i)^2\mu_i g_i^2\Big)
+
O(\mu_i)
+
O\!\Big((s a_i)^2 x_i^{\star 2}\mu_i\Big).
\end{aligned}
\]
Using $\frac{(s a_i)^2}{\mu_i}x_i^{\star 2}=\mu_i$ and $(s a_i)^2 g_i x^\star_i=\mu_i(s a_i g_i)$, this can be rewritten as
\[
\begin{aligned}
\frac{\mathbb E(\hat x_{\text{HG},i}(y_i)-x^\star_i)^2}{\mathbb E(\hat x_{\mathrm{MLE,P},i}-x^\star_i)^2}
=
\mu_i
+
(s a_i)^2 g_i^2
-
2\mu_i(s a_i g_i) +
O\!\Big(\mu_i (s a_i)^2 g_i^2\Big)
+
O(\mu_i)
+
O(\mu_i^3).
\end{aligned}
\]
Since $0\le s a_i g_i\le 1$ and $0\le (s a_i)^2 g_i^2\le 1$ (because $g_i\le 1/(s a_i)$), the last expansion simplifies to
\[
\begin{aligned}
\frac{\mathbb E(\hat x_{\text{HG},i}(y_i)-x^\star_i)^2}{\mathbb E(\hat x_{\mathrm{MLE,P},i}-x^\star_i)^2}
&=
(s a_i)^2 g_i^2 + O(\mu_i),
\qquad \mu_i\to 0.
\end{aligned}
\]
As we discussed in the proof of Proposition \ref{prop:per-mode-mse-pmle}, this statement alone does not identify the leading order unless one compares the sizes of $(s a_i)^2 g_i^2$ and $\mu_i$.
More precisely, the expansion
\[
\frac{\mathbb E(\hat x_{\text{HG},i}(y_i)-x^\star_i)^2}{\mathbb E(\hat x_{\mathrm{MLE,P},i}-x^\star_i)^2}
=
(s a_i)^2 g_i^2 + O(\mu_i)
\]
implies
\[
\frac{\mathbb E(\hat x_{\text{HG},i}(y_i)-x^\star_i)^2}{\mathbb E(\hat x_{\mathrm{MLE,P},i}-x^\star_i)^2}
\sim (s a_i)^2 g_i^2
\]
if and only if
\[ 
\mu_i = o\!\big((s a_i)^2 g_i^2\big), \qquad \mu_i \to 0.
\]

We can express $(s a_i)^2 g_i^2$ explicitly. Define
\[
\begin{aligned}
c(\epsilon)
:=
(s a_i)g_i
&=
\left(
\frac{-1+\sqrt{1+4(1+\epsilon)^2}}{2}-\epsilon
\right),
\end{aligned}
\]
which depends only on $\epsilon>0$; crucially, it does \emph{not} depend on $i$.
Then
\[
(s a_i)^2 g_i^2=c(\epsilon)^2,
\]
and therefore
\[
\frac{\mathbb E(\hat x_{\text{HG},i}(y_i)-x^\star_i)^2}{\mathbb E(\hat x_{\mathrm{MLE,P},i}-x^\star_i)^2}
=
c(\epsilon)^2 + O(\mu_i),
\qquad \mu_i\to 0.
\]

Note that the constant $c(\epsilon)$ satisfies 
\[ 
\frac12 < c(\epsilon) < \frac{\sqrt5-1}{2} \quad \text{for all }\epsilon>0.
\]
Indeed, $\sqrt{1+4(1+\epsilon)^2} \le 1+2(1+\epsilon)$ implies $c(\epsilon)\le 1$, and the sharper inequality
\[
\sqrt{1+4(1+\epsilon)^2}
\le
2(1+\epsilon)+\frac{1}{4(1+\epsilon)}
\]
yields $c(\epsilon)\le \frac12+\frac{1}{8(1+\epsilon)}$, hence $c(\epsilon)\downarrow \frac12$ as $\epsilon\to\infty$.
Moreover, $c(\epsilon)\uparrow \frac{\sqrt5-1}{2}$ as $\epsilon\downarrow 0$.

In particular,
\[
\begin{aligned}
c(\epsilon)
&=
\frac12+\frac{1}{8(1+\epsilon)}
+O\!\Big(\frac{1}{(1+\epsilon)^3}\Big),
\qquad \epsilon\to\infty,
\end{aligned}
\]
so that
\[
\begin{aligned}
c(\epsilon)^2
&=
\frac14+\frac{1}{8(1+\epsilon)}
+O\!\Big(\frac{1}{(1+\epsilon)^2}\Big),
\qquad \epsilon\to\infty.
\end{aligned}
\]
Thus, for fixed $\epsilon>0$, 
\[
\frac{\mathbb E(\hat x_{\text{HG},i}(y_i)-x^\star_i)^2}{\mathbb E(\hat x_{\mathrm{MLE,P},i}-x^\star_i)^2}
\in\left(\frac{1}{4},\frac{3-\sqrt5}{2}\right)
\]
up to the $O(\mu_i)$ error.
\end{proof}

\subsubsection{Global MSE ratio}
We now discuss what Proposition~\ref{prop:per-mode-mse-hgmle} implies for the global MSE.
As in the previous proofs, we use the decomposition
\[
\mathbb E \|\hat \xb_{\mathrm{MLE,P}}(\yb) -\xb^\star\|^2 = V_d + T_d,
\]
where 
\[
\begin{aligned}
V_d &:= \sum_{i=1}^d \frac{x^\star_i}{sa_i}, \qquad
B_d:=\sum_{i=1}^d x_i^{\star 2}, \qquad
T_d:= \sum_{i>d} x_i^{\star 2}.
\end{aligned}
\]
For each $i$, we have
\[
\hat x_{G,i}(1)=\frac{c(\epsilon)}{s a_i},
\]
hence $(s a_i)^2\,\hat x_{\text{HG},i}(1)^2 = c(\epsilon)^2$. 
Let $\mu_i \to 0$; from the mode-by-mode low-dose expansion,
\[
\frac{\mathbb E(\hat x_{\text{HG},i}(y_i)-x^\star_i)^2}{\mathbb E(\hat x_{\mathrm{MLE,P},i}-x^\star_i)^2}
=
c(\epsilon)^2 + O(\mu_i),
\]
and since $\mathbb E(\hat x_{\mathrm{MLE,P},i}-x^\star_i)^2 = x^\star_i/(s a_i)$, this is equivalent to
\[
\begin{aligned}
\mathbb E(\hat x_{\text{HG},i}(y_i)-x^\star_i)^2
&=
c(\epsilon)^2\frac{x^\star_i}{s a_i}
+
r_i,\\
r_i
&= O\!\Big(\mu_i\,\frac{x^\star_i}{s a_i}\Big).
\end{aligned}
\]
Using $\mu_i\,\frac{x^\star_i}{s a_i}=x_i^{\star 2}$, the remainder can be written as
\[
r_i = O(x_i^{\star 2}).
\]
More explicitly, from the MSE  expansion
\[
\begin{aligned}
\mathbb E(\hat x_{\text{HG},i}(y_i)-x^\star_i)^2
=
x_i^{\star 2}
+\mu_i\big(g_i^2-2g_i x^\star_i\big) +
O\!\Big(
\mu_i^2 g_i^2
+ \frac{\mu_i^2}{(s a_i)^2}
+ x_i^{\star 2}\mu_i^2
\Big),
\end{aligned}
\]
with $g_i=\hat x_{\text{HG},i}(1)=c(\epsilon)/(s a_i)$, we obtain
\[
\begin{aligned}
\mathbb E(\hat x_{\text{HG},i}(y_i)-x^\star_i)^2 =
c(\epsilon)^2\frac{x^\star_i}{s a_i}
+
\bigl(1-2c(\epsilon)\bigr)x_i^{\star 2}
+
O(x_i^{\star 2})
+
O\bigl((s a_i)^2 x_i^{\star 4}\bigr), 
\end{aligned}
\]
as $\mu_i\to 0$. When passing to the global MSE, we sum these remainders over $i\le d$; under the uniform low-dose condition
$\max_{i\le d(s)}\mu_i\to 0$, the constants implicit in the $O(\cdot)$ terms may be chosen uniformly over $i\le d(s)$, so that the summed remainders are
$O(B_d)$ and $O\!\big(\sum_{i=1}^d (s a_i)^2 x_i^{\star 4}\big)$.

Summing over $i\le d$ yields the global expansion
\[
\begin{aligned}
\mathbb E\big\|\hat \xb_\text{HG}(\yb)-\xb^\star\big\|^2
=
c(\epsilon)^2 V_d
+
\bigl(1-2c(\epsilon)\bigr) B_d
+
T_d +
O(B_d)
+
O\!\Big(\sum_{i=1}^d (s a_i)^2 x_i^{\star 4}\Big).
\end{aligned}
\]

Dividing by the MSE relative to Poisson MLE gives
\[
\begin{aligned}
\frac{\mathbb E\|\hat \xb_\text{HG}(\yb)-\xb^\star\|^2}
     {\mathbb E\|\hat \xb_{\mathrm{MLE,P}}(\yb)-\xb^\star\|^2}
=
\frac{
c(\epsilon)^2 V_d
+\bigl(1-2c(\epsilon)\bigr) B_d
+T_d
+
O(B_d)
+ O\!\big(\sum_{i=1}^d (s a_i)^2 x_i^{\star4}\big)
}{V_d+T_d}.
\end{aligned}
\]

In particular, if $d=d(s)$ is chosen so that 
\[
\begin{aligned}
\frac{B_{d(s)}}{V_{d(s)}}\to 0,
\qquad
\frac{T_{d(s)}}{V_{d(s)}}\to 0,\qquad \frac{\sum_{i=1}^{d(s)} (s a_i)^2 x_i^{\star4}}{V_{d(s)}}\to 0,
\qquad s\to 0,
\end{aligned}
\]
and (uniform per-mode low-dose regime)
\[
\begin{aligned}
\max_{1\le i\le d(s)} \mu_i
&=
\max_{1\le i\le d(s)} s a_i x^\star_i
\to 0,
\qquad s\to 0, 
\end{aligned}
\]
then the global MSE ratio converges to the same constant as the mode-by-mode ratio:
\[
\begin{aligned}
\frac{\mathbb E\|\hat \xb_\text{HG}(\yb)-\xb^\star\|^2}{\mathbb E\|\hat \xb_{\mathrm{MLE,P}}(\yb)-\xb^\star\|^2}
&\to
c(\epsilon)^2
\in\left(\frac{1}{4},\frac{3-\sqrt5}{2}\right),
\qquad s\to 0.
\end{aligned}
\]

\def\ps{0.45}
\begin{figure}[ht!]
\centering
    \centering
    \begin{tikzpicture}
    
    \pgfplotsset{
      every axis/.append style={
        tick label style={font=\scriptsize},
        label style={font=\scriptsize},
        xlabel style={yshift=2pt},
        ylabel style={yshift=-1pt},
        ticklabel style={inner sep=1pt},
        enlarge x limits=0.03,
        enlarge y limits=0.08,
        clip mode=individual,
      }
    }

    \begin{groupplot}[
      group style={
        group size=2 by 1,
        horizontal sep=0.85cm,
        vertical sep=0.42cm
      },
      every axis plot/.append style={mark options={solid}
      },
      width=0.5\textwidth,
      height=0.35\textwidth,
      xmode=log,
      log basis x=10,
      grid=both,
      major grid style={dashed,gray!30},
      minor grid style={dotted,gray!10},
      minor tick num=9,
      xtick={0.3,1,3,10,30,100,300,1000},
      xticklabels={0.3,1,3,10,30,100,300,1000},
      xtick pos=lower,
      ytick pos=left,
      tick align=outside,
      xlabel near ticks,
      ylabel near ticks,
      xlabel={\scriptsize Average expected number of counts $c$},
      ymode=log,
      ymin=7*1e-3,
      ymax=6e-2,
      ytick={0.01,0.02,0.03,0.05},
      yticklabels={$1.10^{-2}$,$2.10^{-2}$,$3.10^{-2}$,$5.10^{-2}$},
      ytickten={-3,-2,-1},
      legend cell align={left},
      legend style={
      at={(axis description cs:0.98,0.98)},
      anchor=north east,
      font=\scriptsize,
      draw=none,
      fill=white,
      fill opacity=0.9,
      text opacity=1,
      inner sep=2pt,
    }
    ]

\nextgroupplot[
  ylabel={\scriptsize MSE},
  title style={
    at={(rel axis cs:0.98,0.90)},
    anchor=north east,
    font=\normalsize,
    inner sep=1pt,
  },
]

\addplot+[OIvermillion, solid, line width=1.0pt,
          mark=diamond*, mark size=2.6pt,
          mark options={solid, draw=OIvermillion, fill=OIvermillion}]
  table[x=target_c,y=HG, col sep=comma]{csv_files_for_plots/shepp/mse_vs_tc-shepp.csv};
\addlegendentry{Regularized HG MAP}

\addplot+[OIgreen, densely dashdotted, line width=1.0pt,
          mark=square*, mark size=2.6pt,
          mark options={solid, draw=OIgreen, fill=OIgreen}]
  table[x=target_c,y={WLS (oracle)}, col sep=comma]{csv_files_for_plots/shepp/mse_vs_tc-shepp.csv};
\addlegendentry{PWLS (oracle)}

\addplot+[OIsky, densely dashed, line width=1.0pt,
          mark=triangle*, mark size=2.8pt,
          mark options={solid, draw=OIsky, fill=OIsky}]
  table[x=target_c,y={WLS (plug-in)}, col sep=comma]{csv_files_for_plots/shepp/mse_vs_tc-shepp.csv};
\addlegendentry{PWLS (plug-in)}

\addplot+[OIorange, densely dotted, line width=1.0pt,
          mark=triangle, mark size=3.0pt,
          mark options={solid, draw=OIorange, fill=white, line width=0.6pt}]
  table[x=target_c,y={WLS (plug-in FBP)}, col sep=comma]{csv_files_for_plots/shepp/mse_vs_tc-shepp.csv};
\addlegendentry{PWLS (plug-in FBP)}

\addplot+[OIpurple, dashed, line width=1.0pt,
          mark=triangle*, mark size=3.2pt,
          mark options={solid, rotate=180, draw=OIpurple, fill=OIpurple}]
  table[x=target_c,y=Homoscedastic, col sep=comma]{csv_files_for_plots/shepp/mse_vs_tc-shepp.csv};
\addlegendentry{Homoscedastic LS}

\addplot+[OIblue, solid, line width=1.0pt,
          mark=*, mark size=2.1pt,
            opacity=0.9,
          mark options={solid, draw=OIblue, fill=OIblue}]
  table[x=target_c,y=Poisson, col sep=comma]{csv_files_for_plots/shepp/mse_vs_tc-shepp.csv};
\addlegendentry{Poisson MAP}
\end{groupplot}
\end{tikzpicture}
\caption{MSE against the average expected number of counts per detector bin. 
For HG and all PWLS variants, the stabilization floor $\varepsilon$ was selected from $\{0.1,\,0.5,\,1.0\}$ by minimizing the tuning-set MSE at the lowest count level; the chosen $\varepsilon$ was then fixed and used for all average expected number of counts.}
\label{fig:shepp}
\end{figure}
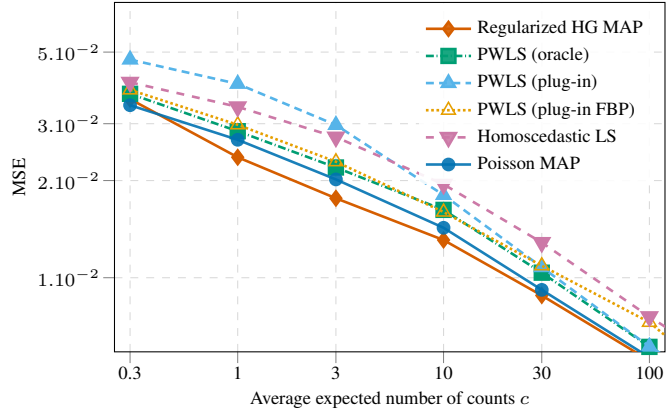

\section{Supplementary Experiment and Sensitivity Analysis}\label{app:suppexp}
Here, we provide additional experiments that support the main conclusions. We report a controlled phantom benchmark, study sensitivity to the stabilization floor $\varepsilon$ used by Gaussian-family objectives, and illustrate the U-shaped MSE curve used to tune the regularization strength. 

\subsection{Quantitative Results on the Shepp--Logan Phantom}\label{app:MSE-shepp-logan}
To remove anatomical variability from the LoDoPaB dataset, we repeat the tomography experiment on a fixed Shepp--Logan phantom. For each average expected number of counts $c$, we generate 10 independent Poisson realizations for tuning and 20 different realizations for testing, and we report mean test-set MSE on the FOV. Figure~\ref{fig:shepp} shows that the qualitative behavior matches the LoDoPaB results.

\begin{figure}[!t]
\centering
\begin{tikzpicture}
    \pgfplotsset{
      every axis/.append style={
        tick label style={font=\scriptsize},
        label style={font=\scriptsize},
        xlabel style={yshift=2pt},
        ylabel style={yshift=-1pt},
        ticklabel style={inner sep=1pt},
        enlarge x limits=0.03,
        enlarge y limits=0.08,
        clip mode=individual,
      }
    }

    \begin{groupplot}[
      group style={
        group size=4 by 1,
        horizontal sep=0.5cm,
        vertical sep=0.42cm
      },
      every axis plot/.append style={mark options={solid}},
      width=0.30\textwidth,
      height=0.24\textwidth,
      xmode=log,
      log basis x=10,
      grid=both,
      major grid style={dashed,gray!30},
      minor grid style={dotted,gray!10},
      minor tick num=9,
      xtick pos=lower,
      ytick pos=left,
      tick align=outside,
      xlabel near ticks,
      ylabel near ticks,
      xlabel={\scriptsize $c$},
      ymode=log,
      ymin=1e-3, 
      ymax=1e-1,
      ytick={1e-1,1e-2,1e-3},
      yticklabels={$10^{-1}$,$10^{-2}$,$10^{-3}$},
      legend cell align={left},
    ]

    \nextgroupplot[
      ylabel={\scriptsize MSE},
      legend to name=sharedlegend,
      legend columns=4,
      legend style={
        draw=black,
        fill=white,
        fill opacity=0.9,
        text opacity=1,
        inner sep=4pt,
        font=\scriptsize,
        column sep=0.8em,
      },
    ]
    
    \addplot+[OIblue, solid, line width=1.0pt,
              mark=*, mark size=1.8pt,
              opacity=0.9,
              mark options={solid, draw=OIblue, fill=OIblue}] 
      table[x=target_c, y=Poisson, col sep=comma] {csv_files_for_plots/lodo/eps_robustness_hg.csv};
    \addlegendentry{Poisson MAP}
    
    \addplot+[OIorange, solid, line width=1.0pt,
              mark=square*, mark size=2.0pt,
              mark options={solid, draw=OIorange, fill=OIorange}] 
      table[x=target_c, y=eps0.1, col sep=comma] {csv_files_for_plots/lodo/eps_robustness_hg.csv};
    \addlegendentry{$\varepsilon=0.1$}
    
    \addplot+[OIgreen, densely dashed, line width=1.0pt,
              mark=triangle*, mark size=2.2pt,
              mark options={solid, draw=OIgreen, fill=OIgreen}] 
      table[x=target_c, y=eps0.5, col sep=comma] {csv_files_for_plots/lodo/eps_robustness_hg.csv};
    \addlegendentry{$\varepsilon=0.5$}
    
    \addplot+[OIvermillion, densely dotted, line width=1.0pt,
              mark=diamond*, mark size=2.2pt,
              mark options={solid, draw=OIvermillion, fill=OIvermillion}] 
      table[x=target_c, y=eps1.0, col sep=comma] {csv_files_for_plots/lodo/eps_robustness_hg.csv};
    \addlegendentry{$\varepsilon=1.0$}

    \nextgroupplot[yticklabels={}]
    
    \addplot+[OIblue, solid, line width=1.0pt,
              mark=*, mark size=1.8pt,
              opacity=0.9,
              mark options={solid, draw=OIblue, fill=OIblue}] 
      table[x=target_c, y=Poisson, col sep=comma] {csv_files_for_plots/lodo/eps_robustness_wls_oracle.csv};
    
    \addplot+[OIorange, solid, line width=1.0pt,
              mark=square*, mark size=2.0pt,
              mark options={solid, draw=OIorange, fill=OIorange}] 
      table[x=target_c, y=eps0.1, col sep=comma] {csv_files_for_plots/lodo/eps_robustness_wls_oracle.csv};
    
    \addplot+[OIgreen, densely dashed, line width=1.0pt,
              mark=triangle*, mark size=2.2pt,
              mark options={solid, draw=OIgreen, fill=OIgreen}] 
      table[x=target_c, y=eps0.5, col sep=comma] {csv_files_for_plots/lodo/eps_robustness_wls_oracle.csv};
    
    \addplot+[OIvermillion, densely dotted, line width=1.0pt,
              mark=diamond*, mark size=2.2pt,
              mark options={solid, draw=OIvermillion, fill=OIvermillion}] 
      table[x=target_c, y=eps1.0, col sep=comma] {csv_files_for_plots/lodo/eps_robustness_wls_oracle.csv};

    \nextgroupplot[yticklabels={}]
    
    \addplot+[OIblue, solid, line width=1.0pt,
              mark=*, mark size=1.8pt,
              opacity=0.9,
              mark options={solid, draw=OIblue, fill=OIblue}] 
      table[x=target_c, y=Poisson, col sep=comma] {csv_files_for_plots/lodo/eps_robustness_wls_plugin.csv};
    
    \addplot+[OIorange, solid, line width=1.0pt,
              mark=square*, mark size=2.0pt,
              mark options={solid, draw=OIorange, fill=OIorange}] 
      table[x=target_c, y=eps0.1, col sep=comma] {csv_files_for_plots/lodo/eps_robustness_wls_plugin.csv};
    
    \addplot+[OIgreen, densely dashed, line width=1.0pt,
              mark=triangle*, mark size=2.2pt,
              mark options={solid, draw=OIgreen, fill=OIgreen}] 
      table[x=target_c, y=eps0.5, col sep=comma] {csv_files_for_plots/lodo/eps_robustness_wls_plugin.csv};
    
    \addplot+[OIvermillion, densely dotted, line width=1.0pt,
              mark=diamond*, mark size=2.2pt,
              mark options={solid, draw=OIvermillion, fill=OIvermillion}] 
      table[x=target_c, y=eps1.0, col sep=comma] {csv_files_for_plots/lodo/eps_robustness_wls_plugin.csv};

    \nextgroupplot[yticklabels={}]
    
    \addplot+[OIblue, solid, line width=1.0pt,
              mark=*, mark size=1.8pt,
              opacity=0.9,
              mark options={solid, draw=OIblue, fill=OIblue}] 
      table[x=target_c, y=Poisson, col sep=comma] {csv_files_for_plots/lodo/eps_robustness_wls_plugin_fbp.csv};
    
    \addplot+[OIorange, solid, line width=1.0pt,
              mark=square*, mark size=2.0pt,
              mark options={solid, draw=OIorange, fill=OIorange}] 
      table[x=target_c, y=eps0.1, col sep=comma] {csv_files_for_plots/lodo/eps_robustness_wls_plugin_fbp.csv};
    
    \addplot+[OIgreen, densely dashed, line width=1.0pt,
              mark=triangle*, mark size=2.2pt,
              mark options={solid, draw=OIgreen, fill=OIgreen}] 
      table[x=target_c, y=eps0.5, col sep=comma] {csv_files_for_plots/lodo/eps_robustness_wls_plugin_fbp.csv};
    
    \addplot+[OIvermillion, densely dotted, line width=1.0pt,
              mark=diamond*, mark size=2.2pt,
              mark options={solid, draw=OIvermillion, fill=OIvermillion}] 
      table[x=target_c, y=eps1.0, col sep=comma] {csv_files_for_plots/lodo/eps_robustness_wls_plugin_fbp.csv};

    \end{groupplot}
    
    \node[anchor=south] at ($(group c2r1.north)!0.5!(group c3r1.north) + (0,0.2cm)$)
      {\pgfplotslegendfromname{sharedlegend}};
    
    \node[anchor=north, font=\scriptsize] at ($(group c1r1.south) + (0,-0.8cm)$) {(a) HG MAP};
    \node[anchor=north, font=\scriptsize] at ($(group c2r1.south) + (0,-0.8cm)$) {(b) PWLS (oracle)};
    \node[anchor=north, font=\scriptsize] at ($(group c3r1.south) + (0,-0.8cm)$) {(c) PWLS (plug-in)};
    \node[anchor=north, font=\scriptsize] at ($(group c4r1.south) + (0,-0.8cm)$) {(d) PWLS (plug-in FBP)};
      
\end{tikzpicture}
\caption{LoDoPaB-CT: sensitivity to the stabilization floor $\varepsilon$. Each panel shows test MSE versus average expected counts $c$ for a fixed method, comparing $\varepsilon\in\{0.1,0.5,1.0\}$ to the Poisson MAP baseline. The results from the Shepp--Logan phantom experiments exhibits the same qualitative behavior.}
\label{fig:eps_robustness}
\end{figure}
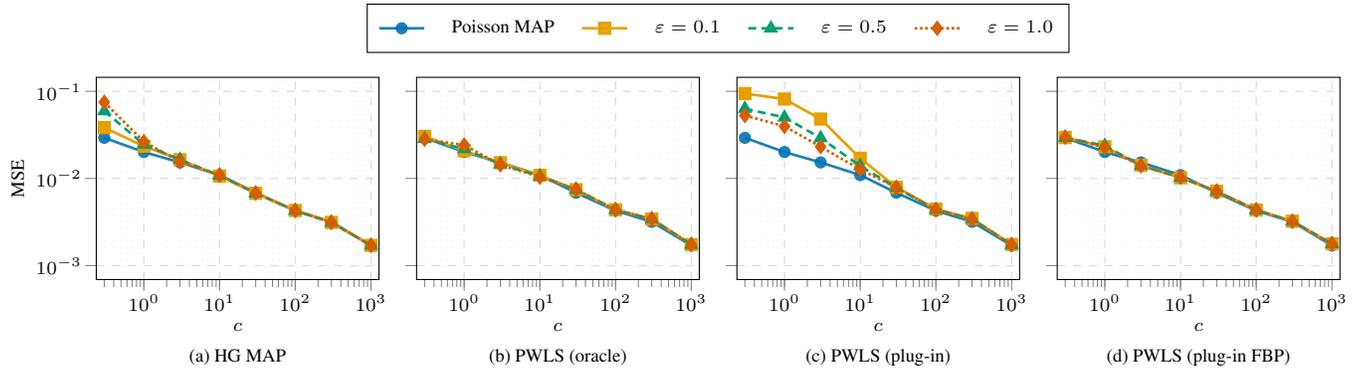

\def\ps{0.45}
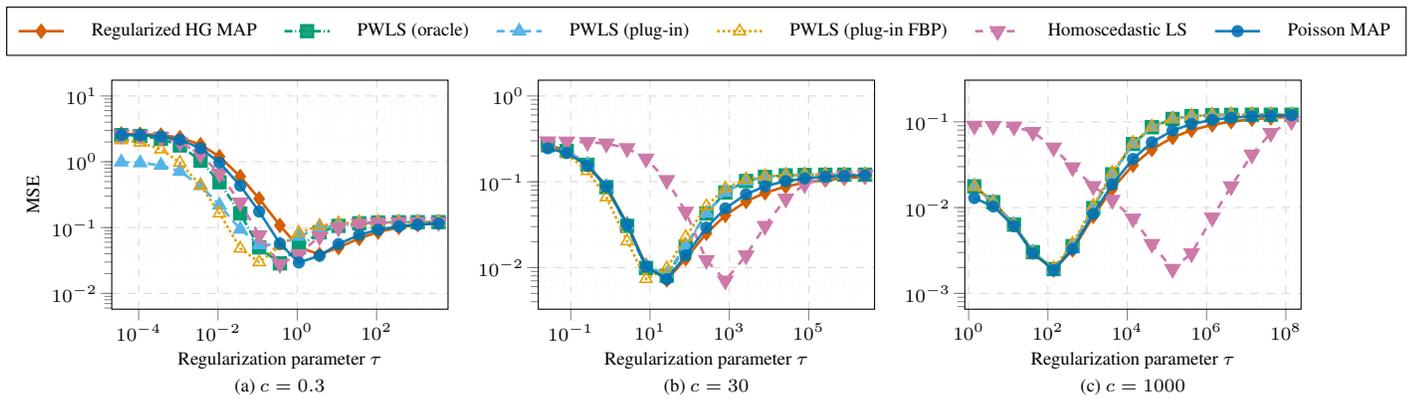
\begin{figure}[!t]
\centering
\begin{tikzpicture}
    
    \pgfplotsset{
      every axis/.append style={
        tick label style={font=\scriptsize},
        label style={font=\scriptsize},
        xlabel style={yshift=2pt},
        ylabel style={yshift=-1pt},
        ticklabel style={inner sep=1pt},
        enlarge x limits=0.03,
        enlarge y limits=0.08,
        clip mode=individual,
      }
    }

    \begin{groupplot}[
      group style={
        group size=3 by 1,
        horizontal sep=1.2cm,
        vertical sep=0.42cm
      },
      every axis plot/.append style={mark options={solid}},
      width=0.34\textwidth,
      height=0.26\textwidth,
      xmode=log,
      log basis x=10,
      grid=both,
      major grid style={dashed,gray!30},
      minor grid style={dotted,gray!10},
      minor tick num=9,
      xtick pos=lower,
      ytick pos=left,
      tick align=outside,
      xlabel near ticks,
      ylabel near ticks,
      xlabel={\scriptsize Regularization parameter $\tau$},
      ymode=log,
      legend cell align={left},
    ]

    \nextgroupplot[
      ylabel={\scriptsize MSE},
      ymin=1e-2,
      ymax=1e1,
      legend to name=mselegend,
      legend columns=6,
      legend style={
        draw=black,
        fill=white,
        fill opacity=0.9,
        text opacity=1,
        inner sep=4pt,
        font=\scriptsize,
        column sep=0.8em,
      },
    ]

    \addplot+[OIvermillion, solid, line width=1.0pt,
              mark=diamond*, mark size=2.2pt,
              mark options={solid, draw=OIvermillion, fill=OIvermillion}]
      table[x=alpha,y=HG, col sep=comma]{csv_files_for_plots/lodo/tune_mse_vs_alpha_tc_0.3.csv};
    \addlegendentry{Regularized HG MAP}

    \addplot+[OIgreen, densely dashdotted, line width=1.0pt,
              mark=square*, mark size=2.2pt,
              mark options={solid, draw=OIgreen, fill=OIgreen}]
      table[x=alpha,y={WLS (oracle)}, col sep=comma]{csv_files_for_plots/lodo/tune_mse_vs_alpha_tc_0.3.csv};
    \addlegendentry{PWLS (oracle)}

    \addplot+[OIsky, densely dashed, line width=1.0pt,
              mark=triangle*, mark size=2.4pt,
              mark options={solid, draw=OIsky, fill=OIsky}]
      table[x=alpha,y={WLS (plug-in)}, col sep=comma]{csv_files_for_plots/lodo/tune_mse_vs_alpha_tc_0.3.csv};
    \addlegendentry{PWLS (plug-in)}

    \addplot+[OIorange, densely dotted, line width=1.0pt,
              mark=triangle, mark size=2.6pt,
              mark options={solid, draw=OIorange, fill=white, line width=0.6pt}]
      table[x=alpha,y={WLS (plug-in FBP)}, col sep=comma]{csv_files_for_plots/lodo/tune_mse_vs_alpha_tc_0.3.csv};
    \addlegendentry{PWLS (plug-in FBP)}

    \addplot+[OIpurple, dashed, line width=1.0pt,
              mark=triangle*, mark size=2.8pt,
              mark options={solid, rotate=180, draw=OIpurple, fill=OIpurple}]
      table[x=alpha,y=Homoscedastic, col sep=comma]{csv_files_for_plots/lodo/tune_mse_vs_alpha_tc_0.3.csv};
    \addlegendentry{Homoscedastic LS}

    \addplot+[OIblue, solid, line width=1.0pt,
              mark=*, mark size=1.8pt,
              opacity=0.9,
              mark options={solid, draw=OIblue, fill=OIblue}]
      table[x=alpha,y=Poisson, col sep=comma]{csv_files_for_plots/lodo/tune_mse_vs_alpha_tc_0.3.csv};
    \addlegendentry{Poisson MAP}

    \nextgroupplot[
      ymin=5e-3,
      ymax=1e0,
    ]

    \addplot+[OIvermillion, solid, line width=1.0pt,
              mark=diamond*, mark size=2.2pt,
              mark options={solid, draw=OIvermillion, fill=OIvermillion}]
      table[x=alpha,y=HG, col sep=comma]{csv_files_for_plots/lodo/tune_mse_vs_alpha_tc_30.csv};

    \addplot+[OIgreen, densely dashdotted, line width=1.0pt,
              mark=square*, mark size=2.2pt,
              mark options={solid, draw=OIgreen, fill=OIgreen}]
      table[x=alpha,y={WLS (oracle)}, col sep=comma]{csv_files_for_plots/lodo/tune_mse_vs_alpha_tc_30.csv};

    \addplot+[OIsky, densely dashed, line width=1.0pt,
              mark=triangle*, mark size=2.4pt,
              mark options={solid, draw=OIsky, fill=OIsky}]
      table[x=alpha,y={WLS (plug-in)}, col sep=comma]{csv_files_for_plots/lodo/tune_mse_vs_alpha_tc_30.csv};

    \addplot+[OIorange, densely dotted, line width=1.0pt,
              mark=triangle, mark size=2.6pt,
              mark options={solid, draw=OIorange, fill=white, line width=0.6pt}]
      table[x=alpha,y={WLS (plug-in FBP)}, col sep=comma]{csv_files_for_plots/lodo/tune_mse_vs_alpha_tc_30.csv};

    \addplot+[OIpurple, dashed, line width=1.0pt,
              mark=triangle*, mark size=2.8pt,
              mark options={solid, rotate=180, draw=OIpurple, fill=OIpurple}]
      table[x=alpha,y=Homoscedastic, col sep=comma]{csv_files_for_plots/lodo/tune_mse_vs_alpha_tc_30.csv};

    \addplot+[OIblue, solid, line width=1.0pt,
              mark=*, mark size=1.8pt,
              opacity=0.9,
              mark options={solid, draw=OIblue, fill=OIblue}]
      table[x=alpha,y=Poisson, col sep=comma]{csv_files_for_plots/lodo/tune_mse_vs_alpha_tc_30.csv};

    \nextgroupplot[
      ymin=1e-3,
      ymax=0.2,
    ]

    \addplot+[OIvermillion, solid, line width=1.0pt,
              mark=diamond*, mark size=2.2pt,
              mark options={solid, draw=OIvermillion, fill=OIvermillion}]
      table[x=alpha,y=HG, col sep=comma]{csv_files_for_plots/lodo/tune_mse_vs_alpha_tc_1000.csv};

    \addplot+[OIgreen, densely dashdotted, line width=1.0pt,
              mark=square*, mark size=2.2pt,
              mark options={solid, draw=OIgreen, fill=OIgreen}]
      table[x=alpha,y={WLS (oracle)}, col sep=comma]{csv_files_for_plots/lodo/tune_mse_vs_alpha_tc_1000.csv};

    \addplot+[OIsky, densely dashed, line width=1.0pt,
              mark=triangle*, mark size=2.4pt,
              mark options={solid, draw=OIsky, fill=OIsky}]
      table[x=alpha,y={WLS (plug-in)}, col sep=comma]{csv_files_for_plots/lodo/tune_mse_vs_alpha_tc_1000.csv};

    \addplot+[OIorange, densely dotted, line width=1.0pt,
              mark=triangle, mark size=2.6pt,
              mark options={solid, draw=OIorange, fill=white, line width=0.6pt}]
      table[x=alpha,y={WLS (plug-in FBP)}, col sep=comma]{csv_files_for_plots/lodo/tune_mse_vs_alpha_tc_1000.csv};

    \addplot+[OIpurple, dashed, line width=1.0pt,
              mark=triangle*, mark size=2.8pt,
              mark options={solid, rotate=180, draw=OIpurple, fill=OIpurple}]
      table[x=alpha,y=Homoscedastic, col sep=comma]{csv_files_for_plots/lodo/tune_mse_vs_alpha_tc_1000.csv};

    \addplot+[OIblue, solid, line width=1.0pt,
              mark=*, mark size=1.8pt,
              opacity=0.9,
              mark options={solid, draw=OIblue, fill=OIblue}]
      table[x=alpha,y=Poisson, col sep=comma]{csv_files_for_plots/lodo/tune_mse_vs_alpha_tc_1000.csv};

    \end{groupplot}
    
    \node[anchor=south] at ($(group c2r1.north) + (0,0.2cm)$)
      {\pgfplotslegendfromname{mselegend}};
    
    \node[anchor=north, font=\scriptsize] at ($(group c1r1.south) + (0,-0.8cm)$) {(a) $c = 0.3$};
    \node[anchor=north, font=\scriptsize] at ($(group c2r1.south) + (0,-0.8cm)$) {(b) $c = 30$};
    \node[anchor=north, font=\scriptsize] at ($(group c3r1.south) + (0,-0.8cm)$) {(c) $c = 1000$};
      
\end{tikzpicture}
\caption{LoDoPaB-CT: MSE as a function of the Tikhonov regularization parameter $\tau$ for three representative average expected counts $c$. We observe qualitatively similar regularization curves on the Shepp--Logan phantom.}
\label{fig:mse_vs_alpha}
\end{figure}

\subsection{Sensitivity to the stabilization floor $\varepsilon$}\label{app:eps_sensitivity}
The heteroscedastic Gaussian and PWLS objective depend on a stabilization floor $\varepsilon$ that appears in the denominator of the variance or weights. To assess robustness, we run each method with $\varepsilon \in \{0.1, 0.5, 1.0\}$ across all count levels. We report the per-$\varepsilon$ MSE curves in Figure~\ref{fig:eps_robustness}. PWLS with oracle weights and PWLS with plug-in FBP weights are essentially unchanged as 
$\varepsilon$ varies, whereas HG MAP and PWLS with plug-in weights show noticeable sensitivity at the lowest count levels. This indicates that the choice of $\varepsilon$ can matter in the low-dose regime.

\subsection{MSE versus Tikhonov regularization strength} \label{app:u_shaped}
Figure~\ref{fig:mse_vs_alpha} plots MSE as a function of the regularization parameter $\tau$ for three representative count levels $c$. The curves exhibit the expected U-shape: too little regularization leads to noise amplification, while too much regularization introduces bias through oversmoothing.

\end{document}